\documentclass[11pt,a4paper]{article}

\usepackage{amssymb}
\usepackage{amsthm}
\usepackage{amsmath}

\usepackage{lineno}
\usepackage{color}
\usepackage{soul}

\usepackage{hyperref} %
\usepackage{geometry}
\usepackage{xspace}
\newcommand{\bw}{\ensuremath{\mathbf{w}}\xspace}
\newcommand{\bs}{\mathbf{s}}
\newcommand{\bu}{\mathbf{u}}

\newcommand{\bM}{\mathbf{M}}
\newcommand{\SbLSP}{S_{bLSP}}
\newcommand{\cS}{{\cal S}}
\newcommand{\cSbal}{\cS_{bal}}
\newcommand{\cSsturm}{\cS_{Sturm}}

\DeclareMathOperator{\stab}{Stab}
\DeclareMathOperator{\subst}{Subst}
\DeclareMathOperator{\first}{first}
\DeclareMathOperator{\last}{last}
\DeclareMathOperator{\stabfin}{StabFin}
\DeclareMathOperator{\stablet}{StabLet}
\DeclareMathOperator{\genstabfin}{GenStabFin}
\DeclareMathOperator{\adic}{adic}
\DeclareMathOperator{\id}{Id}
\DeclareMathOperator{\stabultlet}{StabUltLet}
\DeclareMathOperator{\exch}{Exch}

\newtheorem{thm}{Theorem}[section]
\newtheorem{theorem}[thm]{Theorem}
\newtheorem{proposition}[thm]{Proposition}
\newtheorem{lemma}[thm]{Lemma}

\newtheorem{corollary}[thm]{Corollary}

\newtheorem{remark}[thm]{Remark}

\begin{document}


\title{On sets of indefinitely desubstitutable words}


\author{Gwena\"el Richomme\\LIRMM,Université Paul-Valéry Montpellier 3,\\ Université de Montpellier, CNRS, Montpellier, France}

\maketitle

\begin{abstract}
The stable set associated to a given set ${\cal S}$ of nonerasing endomorphisms or substitutions
is the set of all right infinite words that can be indefinitely desubstituted over ${\cal S}$.
This notion generalizes the notion of sets of fixed points of morphisms.
It is linked to $S$-adicity and to property preserving morphisms. 
Two main questions are considered.
Which known sets of infinite words are stable sets?
Which ones are stable sets of a finite set of substitutions?
While bringing answers to the previous questions, some new characterizations  of several
well-known sets of words {such} as the set of binary balanced words or the set of episturmian words
are presented.
A characterization of the set of nonerasing endomorphisms that preserve episturmian words
is also provided.
\end{abstract}



\section{Introduction}

As explained with more details in, for instance,  \cite{Berthe_Delecroix2014RIMS,Pytheas2002}, 
the terminology $S$-adic was introduced by S. Ferenczi {in} \cite{Ferenczi1996ETDS}{,}
where he proved that dynamical symbolic systems
with subaffine factor complexity are $S$-adic 
uniformly minimal symbolic systems.
This result was motivated by the so-called $S$-adic conjecture:
there exists a stronger notion of $S$-adicity which is equivalent to linear factor complexity.
Some advances on this conjecture were obtained by J.~Leroy \cite{Leroy2012thesis,Leroy2012AAM}{,}
especially when {the first difference of} the factor complexity is bounded by 2  \cite{Leroy2014DMTCS}.
Also an adaptation for infinite words was obtained \cite{Leroy_Richomme2012integers}.
Many classical words are known to be $S$-adic words {such} as, \textit{e.g.}, 
fixed points of morphisms, morphic words (images of fixed points), Sturmian words, 3-interval exchange transformations, Arnoux-Rauzy words {and} strict episturmian words. 

In \cite{Berthe_Delecroix2014RIMS}, Berthé and Delecroix write: ``\textit{Expansions of S-adic nature have now proved their efficiency
for yielding convenient descriptions for highly structured symbolic dynamical systems [...]\
If one wants to understand such a system [...],
it might prove to be convenient to decompose it via a desubstitution process: 
an $S$-adic system is a system that can be indefinitely desubstituted}."
The aim of the paper is to study more specifically this desubstitutive process 
under the approach of limit points as defined by P.~Arnoux, M.~Mizutani and T.~Sella \cite{ArnouxMizutaniSellami2014TCS} (and also used, for instance, by Justin {et} al.~\cite{DroubayJustinPirillo2001,JustinPirillo2002}). 
More specifically we study stable sets as defined by E.~Godelle \cite{Godelle2010AAM}, that is, sets of limit points of a given set of substitutions (or free monoid nonerasing endomorphisms). In other terms, a stable set
associated to a set $\cS$ of substitutions
is the set of all right infinite words 
that can be indefinitely desubstituted using elements of $\cS$.

In \cite{Richomme2019IJFCS},
answering a question of G.~Fici, 
the author characterizes in term{s} of limit $S$-adicity
the family of so-called LSP infinite words,
that is, the words having all their left special factors as prefixes.
For this{,} he determines a suitable set $\SbLSP$ of morphisms 
and an automaton recognizing the allowed infinite sequences of desubstitutions. 
As the obtained characterization is quite {involved},
a second part of \cite{Richomme2019IJFCS} considers the question of finding a simpler limit $S$-adic characterization.
Unfortunately, there exists no set $\cS$ of endomorphisms 
such that the set of LSP infinite words is (exactly) the stable set associated to $\cS$
except in the binary case. The main motivation of our study is the question{:} 
which are the known families of infinite words defined by a combinatorial property $P$ 
that correspond to stable sets? 
In \cite{Richomme2019IJFCS}, it was observed that, 
when such a situation arises, morphisms of $\cS$ necessarily preserve the property $P$ of infinite words.
While bringing answers to the previous question, we present some new characterizations of several
families of words {such} as the set of balanced words or the set of episturmian words. 
This leads us also to characterize the set of nonerasing endomorphisms that preserve episturmian words.

In Section~\ref{sec:fixed_points},
we present the process of desubstitution as a generalization of the notion of a fixed point of a morphism.
We also recall needed definitions {such} as those of indefinitely desubstituted words, limit points and stable sets.
In Section~\ref{sec:balanced_words}, we consider an example of {a} stable set, 
proving that the set of binary balanced words is a stable set
of a particular set of four morphisms{.}
As far as we know, this result was not stated formally earlier{,} even if many aspects of the proof are known. 
The considered process of desubstitution is infinite and 
an infinite sequence of substitutions of a given set $\cS$
is associated to the infinite desubstitution of an infinite word:
such a sequence of substitutions is called a directive sequence.
We show that, in the general case, any infinite sequence of elements of $\cS$ is the directive sequence
of at least one infinite word. 

In Section~\ref{sec:structure}, we  study the possible forms of the elements of a stable set
in relation with $S$-adicity.
We show that the characterization of the set of  balanced words can be transformed to another characterization in term{s} of $S$-adicity even if a general transformation cannot exist for arbitrary stable sets.
In Section~\ref{sec:other_binary_examples}, we provide two more examples of 
sets of binary infinite words that are stable sets:
the set of Sturmian words and the set of Lyndon Sturmian words.
For both sets, there exist only infinite sets of substitutions for which the set is a stable set.

In Section~\ref{sec:episturmian},
we consider episturmian words on arbitrary alphabets.
As far as we know, the fact that the set of standard episturmian word is a stable set is the unique result
{of} this form that was previously stated (without adapted terminology) \cite{JustinPirillo2002}. 
After recalling this result, we prove that the set of $A$-strict standard episturmian words,
the set of episturmian words and the set  of $A$-strict episturmian words
are stable sets{,} but only of infinite sets of substitutions.
For this{,} we prove and use a characterization of endomorphisms preserving episturmian words.
Observe that in \cite{JustinPirillo2002} there exists a characterization of episturmian words
using a desubstitution process associated to a finite set $S$ of substitutions 
but not all elements of the stable set of $S$ are episturmian (only the recurrent ones are).

\section{\label{sec:fixed_points}From fixed points to stable sets}

We assume that readers are familiar with combinatorics on words; for omitted definitions see, \textit{e.g.}, \cite{Berthe_Rigo2010CANT,Lothaire1983book,Lothaire2002}. 
All the infinite words considered in this paper are right infinite words. Let us recall some basic notions on fixed points of morphisms.

{Let $A$ be a (finite) alphabet. Let $\#A$ be its cardinality.
The set of words over $A$,} usually denoted $A^*$, equipped with the concatenation operation has a free monoid structure with
neutral element the empty word $\varepsilon$. Given two alphabets $A$ and $B$, a (free monoid) \textit{morphism} $f$ is a map from $A^*$ to $B^*$
that preserves the monoid structure: for all words $u$ and $v$, $f(uv) = f(u)f(v)$ 
(and, consequently, $f(\varepsilon) = \varepsilon$).
Morphisms are completely defined by images of letters.
In what follows we will essentially use \textit{endomorphisms}, that is, morphisms from a free monoid $A^*$ to itself. 

Given an alphabet $A$, the set of right infinite words, that is, infinite sequences of elements of $A$, is usually denoted $A^\omega$. Images by a morphism of infinite words are defined naturally but to preserve infinity, 
one may consider only \textit{nonerasing} morphisms (images of nonempty words are never the empty word). 
As often done, such a nonerasing morphism will be called a \textit{substitution}.

A \textit{fixed point} of an endomorphism is any finite or infinite word $w$ such that $w = f(w)$. 
For instance, the following morphism has both finite and infinite fixed points. Actually{,} since it admits a finite fixed point, it has infinitely many finite and infinitely many infinite fixed points. 

$$f_1:\left\{\begin{array}{l}a \mapsto a\\b \mapsto bac\\c \mapsto baca\end{array}\right.$$

The definition of fixed points is not constructive. To determine more precisely prefixes of an infinite fixed point, at least two different approaches are usually considered.

\subsection*{Limit approach}

If $\bw$ is an infinite fixed point of a morphism $f$ and  if $p$ is one of its prefixes, 
then necessarily, for $n \geq 0$, the word $f^n(p)$ is also a prefix of $\bw$ 
(classically $f^n = f^{n-1} \circ f$, $f^0$ being the identity morphism).  
If $\lim_{n \to \infty} |f^n(p)| = \infty$ (where $|u|$ denotes the length, 
that is, the number of letters, of the word $u$), 
then the words $f^n(p)$ provide progressively all the letters of $\bw$ that,
thus, can be viewed as a limit of finite words. 
This is usually denoted $\bw = \lim_{n \to \infty} f^n(p)$ or $f^\omega(p)$ (in this last case, $p$ is very often a letter).

Observe that the morphism $f_1$ has exactly one fixed point, the word $a^\omega$, 
that cannot be obtained as a limit of images of powers of $f_1$ applied to a word.
The other fixed points are the words $\lim_{n \to \infty} f_1^n( a^kb)$ with $k \geq 0$ an integer.

\subsection*{Desubstitutive approach}

The second approach consists {of} firstly considering 
the fixed point as the image of another word by the morphism.
There exists a word $\bw'$ such that 
$\bw = f(\bw')$: $\bw'$ is called a \textit{desubstituted} word. In general{,} $\bw'$ may not be unique. For instance, the unique fixed point of $f_1$ starting with the letter $b$ can be seen as the image of a word over $\{a, c\}$ as {$f_{1}(c) = f_{1}(ba)$}. Nevertheless in our context, we know that one possibility for $\bw'$ is $\bw$ itself. And so we can iterate the desubstitution. 
Hence we can find an infinite sequence $(\bw_i)_{i \geq 0}$ of infinite words such that $\bw = \bw_0$ and, for all $i \geq 0$, $\bw_i = f(\bw_{i+1})$: the word $\bw$ can be \textit{indefinitely desubstituted}.

For obtaining large prefixes of a fixed point, 
the desubstitutive approach is less constructive than the limit approach but it can be used as follow{s}. 
Consider the first letter $a_0$ of $\bw$. Then consider all letters whose images by $f$ start with $a_0$: each of {these letters} $a_1$ is a possible letter for the first letter of $\bw_1$ {whenever} $f(a_1)$ is a prefix of $\bw$. Then we can iterate looking for potential first letters $a_2$ of $\bw_2$, $a_3$ of $\bw_3$ and so on.
The words $f^n(a_n)$, with $a_n$ the first letter of $\bw_n$, provide longer and longer prefixes of $\bw$ except if this word starts with a finite fixed point. 
In this case we have to consider the letter occurring after this fixed point and iterate, from this letter, the search for other letters.

It follows {from} the definition that any fixed point can be indefinitely desubstituted (remember that some fixed points cannot be defined by the limit approach).
But, for some morphisms, there may exist some infinite words that can be indefinitely desubstituted without being  fixed points. 
This does not occur with the morphism $f_1$ but {it does occur} with the morphism $f_2$ defined  by $f_2(a) = ba$ and $f_2(b) = ab$
{(Observe that$f_2$ has no fixed point)}. The words that can be indefinitely desubstituted  using $f_2$ 
are the fixed points of $f_2^2$ and not {of} $f_2$ (these words are also the \textit{Thue-Morse words}, that is, the fixed points of the morphism $\mu$ defined by $\mu(a) = ab$ and $\mu(b) = ba$). The previous phenomenon is much more general.

\begin{lemma}[{\cite[Prop. 7.1]{Godelle2010AAM}}]
\label{L:characterising fixed points}
Let $\bw$ be a right infinite word and let $f$ be a {nonerasing} morphi{s}m. 
The word $\bw$ can be indefinitely desubstituted by $f$ if and only if $\bw$ is the fixed point of $f^n$ for some integer $n \geq 1$.
\end{lemma}

{Proposition~7.1 in {\protect\cite{Godelle2010AAM}} is richer than the previous lemma. Its proof depends on a larger context. For the sake of completeness,  we provide a short proof in our context. Let $\first(u)$ denote the first letter of a finite nonempty word or infinite word $u$.
}

\begin{proof}
{
First observe that if $\bw$ is the fixed point of $f^n$, 
then $\bw$ can be indefinitely desubstituted. More precisely, 
as $\bw = f^n(\bw)$,
the desubstituted words are $f^{n-1}(\bw)$, \ldots, $f(w)$, $\bw$, $f^{n-1}(\bw)$, \ldots, $f(w)$, $\bw$ and so on.
}

{
Assume now that $\bw$ can be indefinitely desubstituted by $f$.
Let $(\bw_n)_{n \geq 0}$ be a sequence of desubstituted words: $\bw_0 = \bw$ and, for all $n\geq 0$, $w_n = f(\bw_{n+1})$.
Let $a_n = \first(\bw_n)$ for all $n \geq 0$.
As the alphabet $A$ is finite,
there exist arbitrarily large integers $m$, $n$
such that $0 \leq m < n \leq m+\#A$ and $a_n = a_m$.
Observe that, if $0 < m < n$ and $a_n = a_m$, then $a_{n-1} = a_{m-1}$.
Indeed, as $f$ is not erasing, 
$a_{n-1} =$ 
$\first(\bw_{n-1}) =$
$\first(f(\bw_{n})) =$
$\first(f(a_n)) =$ $\first(f(a_m)) =$ $a_{m-1}$.
Consequently, the sequence $(a_n)_{n \geq 0}$
is periodic with a period smaller than or equals to $\#A$.
Let $\pi$ be this period. 
}

{
Let $a = a_0$ ($a$ is the first letter of $\bw$).
For all $n \geq 0$, $a_{\pi n} = a$.
Since $a$ is the first letter of $\bw = f^\pi(\bw_{\pi})$ and since it is the first letter of $\bw_{\pi}$, 
$a$ is the first letter of $f^\pi(a)$.
If $|f^\pi(a)| > 1$, then $\lim_{n \to \infty} |f^{n\pi}(a)| = \infty$.
For all $n \geq 0$,
since $a$ is the first letter of $\bw_{n\pi}$ and as $\bw = f^{n\pi}(\bw_{n\pi})$,
the word $f^{n\pi}(a)$ is a prefix of $\bw$.
Hence $\bw = \lim_{n \to \infty} f^{n\pi}(a)$.
But, similarly, $\bw_\pi = \lim_{n \to \infty} f^{n\pi}(a)$.
Hence $\bw = f^\pi(\bw)$.
}

{Now assume that $|f^\pi(a)| = 1$. Then $\bw = a \bw'$ with $a$ a letter and $\bw'$ an infinite word
such that: $a = f^\pi(a)$; $\bw'$ can be indefinitely desubstituted using $f$ 
(the desubstituted words are the words $\bw_n'$ obtained from the words $\bw_n$ removing their first letters).}

{Thus, by induction, one can state that one of the two following cases holds.}
\begin{enumerate}
\item {$\bw = a_0 \cdots a_N \bw'$ where $N \geq -1$ is an integer (when $N = -1$, $\bw = \bw'$), 
$\bw'$ is an infinite word and $a_0$, \ldots, $a_N$ are letters such that:}
\begin{itemize}
\item {there exists an integer $p$ such that 
$1 \leq p \leq \#A$ and $\bw' = f^p(\bw')$,}
\item {for all $i$, $0 \leq i \leq N$, there exists an integer $\pi_i$ such that 
$1 \leq \pi_i \leq \#A$ and $a_i = f^{\pi_i}(a_i)$.}
\end{itemize}
\item {$\bw = \prod_{i \geq 0} a_i$ where for all $i \geq 0$, $a_i \in A$ and
there exists an integer $\pi_i$ such that $1 \leq \pi_i \leq \#A$ and $a_i = f^{\pi_i}(a_i)$.}

\end{enumerate}
{In the first case, let $n = gcd(p, \pi_0, \dots, \pi_N)$ and, in the second case, let $n =$ $gcd((\pi_i)_{i \geq 0})$.
In both cases, $\bw = f^n(\bw)$.}
\end{proof}

Before going further studying words that are indefinitely desubstitutable over a set of morphisms, 
let us recall and introduce some useful terminology. 
For any set $X$ of finite words, 
let $X^*$ denote the set of all finite words that can be obtained by concatenation of elements of $X$ (including the empty word),
let $X^+$ denote the set of all finite words that can be obtained by concatenation of at least one elements of $X$
and 
let $X^\omega$ denote the set of all infinite words that can be obtained by concatenation of elements of $X$.
Observe that these notations will be used also with sets of substitutions to denote sets of sequences of substitutions 
(the set of substitutions is then considered as an alphabet).
Each finite sequence $\sigma_1 \sigma_2 \cdots \sigma_k$ of substitutions refers both to the sequence 
and to the substitution obtained by composition of the $\sigma_i$ ($\sigma_1 \circ \sigma_2 \circ \ldots \circ \sigma_k$).
We will use both interpretations alternatively:
the context should be clear when we consider sequences and when we consider substitutions.

Let $\subst(A)$ be the set of substitutions on the alphabet A, that is, the set of nonerasing endomorphisms of the free monoid $A^*$.
Let ${\cal S} \subseteq \subst(A)$.
Let $(\sigma_n)_{n \geq 1} \in \cS^\omega$.
Following P.~Arnoux, M.~Mizutani and T.~Sella \cite{ArnouxMizutaniSellami2014TCS},
a finite or an infinite word $\bw$ over $A$ is a \textit{limit point of the sequence}  $(\sigma_n)_{n \geq 1}$
if there exists a sequence $(\bw_n)_{n \geq 0}$ of infinite words over $A$
such that {$\bw = \bw_0$} and {$\bw_n = \sigma_{n+1}(\bw_{n+1})$} for all $n \geq 0$. In other words $\bw$ can be \textit{indefinitely desubstituted} using successively the substitutions in the sequence $(\sigma_n)_{n \geq 1}$. 
The sequence $(\sigma_n)_{n \geq 1}$ will be called a {\textit{directive sequence}} of $\bw$.

Following Godelle \cite{Godelle2010AAM}, for $\cS \subseteq \subst(A)$, 
the \textit{stable set} of ${\cal S}$, denoted {\textit{$\stab(\cS)$}}, is the set of all limit points of sequence{s} in {$\cS^\omega$}.
Also a set $X$ of infinite words is \textit{stabilized} by $\cS$ if $X = \bigcup_{f \in \cS} f(X)$.

\begin{lemma}[\cite{Godelle2010AAM}]~\label{L:Godelle_Stabilized}
\begin{itemize}
\item $\stab(\cS)$ is stabilized by $\cS$.
\item Any set stabilized by $\cS$  is included in $\stab(\cS)$.
\end{itemize}
\end{lemma}

Thus $\stab(\cS)$ is the large{st} set ({w.r.t.} inclusion) of infinite words stabilized by $\cS$. 
Hence the stable set of a set of substitutions appears as a natural generalization of the set of fixed points of powers of a morphism.
This is coherent with Lemma~\ref{L:characterising fixed points} {which} states that the stable set of a singleton $\{f\}$ is the set of fixed points of the morphisms $f^n$ with $n \geq 1$.

Let $P$ be a property of infinite words.
Let $X$ be the set of all infinite words having this property.
A morphism is {said} to \textit{preserve the property} $P$ or \textit{to preserve the set} $X$ 
if for all elements $\bw$ in $X$, $f(\bw)$ also belongs to $X$.
By definition of a stable set, 
for all $\cS \subseteq \subst(A)$,
for all $\sigma \in \cS$ and $\bw \in \stab(\cS)$,
$\sigma(\bw) \in \stab(\cS)$.
This can be reformulated as the next remark, already made in \cite{Richomme2019IJFCS}, 
that will be useful when proving that  some sets of infinite words cannot be stable sets 
of finite sets of substitutions.

\begin{remark}
\label{Rem:preserve}
If $X = \stab(\cS)$ for some set $\cS$ of substitutions, then all elements of $\cS$ preserve the set $X$.
\end{remark}

\section{\label{sec:balanced_words}Balanced words: an example of \texorpdfstring{{a}}{} stable set}

The desubstitutive approach naturally arises in the study of infinite binary balanced words{,} especially in the study of aperiodic balanced words, that is, the Sturmian words. Here we consider the whole set of infinite binary balanced words and we show Proposition~\ref{P:characterising balanced words} below.
Elements of the proof are rather classical.
They allow to illustrate some techniques of proof 
that will be reused later.

A word on an alphabet $A$ is \textit{balanced} 
if for all letters $a$ in $A$ and 
for all factors $u$ and $v$ with $|u|=|v|$, 
the following inequation holds:
$||u|_a-|v|_a| \leq 1$ (here $|w|_\alpha$ denotes the number of occurrences of the letter $\alpha$ in the word $w$).
The aperiodic, that is, {not} ultimately periodic, infinite binary balanced words are called \textit{Sturmian words}. 
The most famous one is the \textit{Fibonacci word}
which is the fixed point of the morphism $\varphi$ defined by $\varphi(a) = ab$, $\varphi(b) = a$.
Since Morse and Hedlund's work \cite{MorseHedlund1940},
periodic infinite binary balanced words are known to be infinite repetitions of conjugates of finite standard words (\textit{e.g.}, $ab(abaab)^\omega = (ababa)^\omega$) (two words $x$ and $y$ are \textit{conjugates} if  there exists a word $u$ such that $xu = uy$). But there are also some ultimately periodic balanced words that are not purely periodic as for instance $a^nba^\omega$, $(ab)^na(ab)^\omega$, $(abaab)^naba(abaab)^\omega$, \ldots 

It is well-known that the following four morphisms are intrinsically related to binary balanced words as they naturally {occur in the study} of Sturmian words. 

\centerline{
{
$L_a: \left\{
\begin{tabular}{l}
$a \mapsto a$\\
$b \mapsto ab$
\end{tabular}\right.$
}
{
$L_b: \left\{
\begin{tabular}{l}
$a \mapsto ba$\\
$b \mapsto b$
\end{tabular} \right.$
}
{
$R_a: \left\{
\begin{tabular}{l}
$a \mapsto a$\\
$b \mapsto ba$
\end{tabular} \right.$
}
{
$R_b: \left\{
\begin{tabular}{l}
$a \mapsto ab$\\
$b \mapsto b$
\end{tabular} \right.$
}
}

Let $\cSbal = \{L_a, R_a, L_b, R_b\}$.
\begin{proposition}\label{P:characterising balanced words}
The set of binary balanced infinite words is the stable set of ${\cal S}_{{\rm  bal}}$.
\end{proposition}

The previous {proposition} means that an infinite binary  word $\bw$ is balanced if and only if there 
exists some words $(\bw_i)_{ i \geq 0}$ 
and morphisms $(\sigma_i)_{i \geq 1}$ in $\{L_a, L_b, R_a, R_b\}$
such that 
{$\bw_0 = \bw$} and for all $i \geq 0$, $\bw_i = \sigma_{i+1}(\bw_{i+1})$. 
One can also see that all words in $(\bw_i)_{i \geq 0}$ are balanced words.
The proof of the \textit{only if} part illustrates the mechanism of desubstitution.

\begin{proof}[Proof of the only if part]
Given an infinite binary balanced word $\bw$ over $\{a, b\}$, 
one of the two words $aa$ and $bb$ does not occur in $\bw$.
In each case and whatever is the first letter of $\bw$, $\bw = \sigma_1(\bw_1)$ 
for a morphism $\sigma_1$ in $\{L_a, L_b, R_a, R_b\}$ 
and an infinite word $\bw_1$ over $\{a, b\}$. The following table summarises the four possible cases:
\begin{center}
\begin{tabular}{|c|c|p{5cm}|}
\hline
$\bw$ starts with & $\bw$ does not contain & $\bw$ can be decomposed  \\
\hline
$a$ &  $bb$ &
over $\{a,ab\}$:
$\bw = L_a(\bw_1)$
\\
\hline
$a$ & $aa$ & 
over $\{ab, b\}$:
$\bw = R_b(\bw_1)$\\
\hline
$b$ & $bb$ & 
over $\{ba,a\}$:
$\bw = R_a(\bw_1)$\\
\hline
$b$ & $aa$ &
over $\{ba, b\}$:
$\bw = L_b(\bw_1)$\\
\hline
\end{tabular}
\end{center}

Let us now state that the word $\bw_1$ is balanced. 
Assume by contradiction that $\bw_1$ is not balanced. 
Considering a pair of words of same minimal length 
such that the first one contains at least two more occurrences of the letter $a$ than the second, 
we find a word $u$ such that both words $aua$ and $bub$ occur in $\bw_1$. 
In the case where $\bw$ starts with $b$ and does not contain $bb$, we verify that the word $\bw = R_a(\bw_1)$ contains  
$aaR_a(u)a$ and $baR_a(u)b$ as factors (the factor $aua$ is not a prefix of $\bw_1$, and so, images of its occurrences in $R_a(\bw_1)$
are necessarily preceded by $a$). 
A similar contradiction with the balance property of $\bw$ can be obtained in the three other cases.

As $\bw_1$ is balanced, iterating what precedes, 
one can find a balanced infinite binary word $\bw_2$ 
and a morphism $\sigma_2$ in $\{L_a, L_b, R_a, R_b\}$ such that $\bw_1 = \sigma_2(\bw_2)$.
The proof of the only if part ends by induction.
\end{proof}

The proof of the \textit{if} part needs the following well-known property:

\begin{lemma}[\mbox{see, \textit{e.g.}, \cite[chap. 2]{Lothaire2002}}]
\label{morphisms preserving balance property}
Any morphism in $\{L_a, L_b, R_a, R_b\}^*$ preserves the balanced property.
\end{lemma}

\begin{proof}[Proof of the if part of Proposition~\ref{P:characterising balanced words}]
Assume that $\bs = (\sigma_n)_{n \geq 1}$ is a sequence of morphisms in $\{L_a,$ $L_b, R_a, R_b\}$ occurring in an infinite desubstitution of a word $\bw$, and let $(\bw_n)_{n \geq 0}$ be the corresponding sequence of words: $\bw_0 = \bw$, $\bw_i = \sigma_{i+1}(\bw_{i+1})$ for all $i \geq 0$.

Assume first that $\bs \in \{L_a, R_a\}^\omega$.
Assume that there exists an integer $n$ such that $ba^nb$ is a factor of $\bw$. 
Thus, by induction, one can see that $ba^{n-i}b$ is a factor of $\bw_i$ for all $i$ such that $0 \leq i \leq n$. Especially $bb$ is a factor of $\bw_n$. This is impossible as it is an image of a word by $L_a$ or $R_a$. 
Thus $\bw$ contains at most one $b$, that is, it belongs to the set
$\{ a^\omega, a^nba^\omega \mid n \geq 1\}$: it is clearly balanced. 

Assume now that more generally, $\bs$ contains finitely many elements of $\{L_b, R_b\}$. 
This means that $\bw = f(\bw')$ with $f \in  \{L_a, L_b, R_a, R_b\}^*$ and $\bw'$ as in the previous case. As $\bw'$ is balanced and as $f$ preserves balanced words by Lemma~\ref{morphisms preserving balance property}, the word $\bw$ is also balanced.

Similarly $\bw$ is balanced if $\bs$ contains finitely many elements of $\{L_a, R_a\}$. 

It remains to study the case where $\bs$ contains both infinitely many occurrences of elements of $\{L_a, R_a\}$ and
infinitely many occurrences of elements of $\{L_b, R_b\}$.
These words are the Sturmian words \cite{BertheHoltonZamboni2006}
and so they are  balanced.
The proof can be given shortly. 
The {sequence $\bs$} can be contracted to a sequence $\bs'$ of morphisms in $\{L_a, R_a\}^+\{L_b, R_b\} \cup  \{L_b, R_b\}^+\{L_a, R_a\}$. 
All the morphisms $\sigma$ occurring in $\bs'$ {satisfy $|\sigma(a)| \geq 2$ and $|\sigma(b)| \geq 2$}. 
Thus one can obtain the word $\bw$ as the limit of prefixes $\sigma_1'\cdots \sigma_n'(a_n)$ with $\sigma_1'\cdots \sigma_n'$ prefixes of $\bs'$. 
Letters $a_n$ are the first letters of the words associated to the infinite desubstitution of $\bw$ by morphisms in $\sigma'$. As letters are balanced words and as all morphisms $\sigma_i$ ($i \geq 1$) 
{preserve the property of being balanced}
by Lemma~\ref{morphisms preserving balance property}, we obtain the balance property of $\bw$.
\end{proof}

Let us observe, from the previous proof, that any infinite sequence of morphisms over $\cSbal$ is the directive sequence of at least one aperiodic or purely periodic balanced {word}. This property is general.

\begin{proposition}
\label{generalisation Arnoux et al}
Let $\cS \subseteq \subst(A)$. Any sequence in $\cS^\omega$ is the directive sequence of at least one element of $\stab(\cS)$.
\end{proposition}

This proposition (and its proof) is a generalization of a result by P.~Arnoux, M.~Mizutani and T.~Sella \cite[Prop2.1]{ArnouxMizutaniSellami2014TCS}: \textit{Any primitive sequence of substitutions has a finite and nonzero number of limit points}.  Let {us} recall that a sequence of substitutions $(\sigma_n)_{n \geq 0}$ over an alphabet $A$ is \textit{primitive}
 if for all $n\geq 0$, there exists an integer $k$ such that, 
for each letter $a$ in $A$, all letters of $A$ occur in $\sigma_n \cdots \sigma_{n+k}(a)$.
For an arbitrary set of substitutions, there may not exist aperiodic limit points. 
Indeed when $\cS = \{ R_b \}$, {we have} $\stab(\cS) = \{ b^\omega \} \cup \{ b^n a b^\omega \mid n \geq 0\}$.
There may also exist  infinitely many aperiodic limit points.
Indeed let $g$ be the morphism defined by $g(a) = abab$, $g(b) = b$. 
The fixed point $g^\omega(a)$ of $g$ starting with $a$ is aperiodic as it contains all words $ab^na$ (with $n \geq 1$) as factors.
The infinite fixed points of $g$ are $b^\omega$ and the words $b^ng^\omega(a)$ for all $n \geq 0$.

\begin{proof}[Proof of Proposition~\ref{generalisation Arnoux et al}]
Let $(\sigma_n)_{n \geq 1}$ be a sequence in $\cS^\omega$.
For any $n \geq 1$, let $f_n: A \mapsto A$ be the map that sends any letter $a \in A$
 to the first letter of $\sigma_n(a)$. 
The sequence $f_1(f_2 \cdots(f_n(A))\cdots)$ of subsets of $A$ is a decreasing sequence of nonempty subsets of $A$, so their intersection is nonempty. 
Let $X$ be the set of all words $a_0 \cdots a_n$, with $n \geq 0$, such that $a_i = f_{i+1}(a_{i+1})$  for all $i, 0 \leq i < n$.
By Konig's lemma, there exists an infinite sequence of letters $(a_n)_{n \geq 0}$
such that $a_n = f_{n+1}(a_{n+1})$ for all $n \geq 0$.

Consider the sequence of words $(u_n)_{n \geq 0}$ defined by $u_n = \sigma_1\sigma_2 \cdots \sigma_{n} (a_{n})$. 
Since $a_{n}$ is the first letter of $\sigma_{n+1}(a_{n+1})$,  $u_n$ is, by construction, a prefix of $u_{n+1}$.
If the sequence {$(u_n)_{n \geq 0}$} is not ultimately periodic, the limit $\lim_{n \to \infty} u_n$ 
defines an infinite word which is a limit point of the sequence $(\sigma_n)_{n \geq 1}$.
If the sequence $u_n$ is ultimately periodic, then its ultimate value is a finite word $u$. 
Thus the infinite word $u^\omega$ is a limit point of $(\sigma_n)_{n \geq 1}$.
\end{proof}

\section{\label{sec:structure}About the structure of stable sets}

In this section, we study links between stable sets and substitutive-adicity.
We also provide a description of elements of a stable set $\stab(\cS)$ that are not $\cS$-adic.
For a survey on $S$-adicity see, \textit{e.g.}, Berthé and Delecroix \cite{Berthe_Delecroix2014RIMS}.
Let {us} recall that an infinite word $\bw$ is \textit{$S$-adic}
if there exist
a sequence $(\sigma_n)_{n \geq 1}$, $\sigma_n: A_{n+1}^* \to A_{n}^*$, of substitutions and
a sequence of letters $(a_n)_{n \geq 1}$
such that 
$\bw = \lim_{n \to \infty} \sigma_1 \cdots \sigma_n(a_n)$.
The sequence $(\sigma_n)_{n \geq 1}$ is called  a \textit{directive sequence} of $\bw$.

Denoting $\cS = \{ \sigma_n \mid n \geq 1\}$, $\bw$ is $\cS$-adic (where $\cS$ refers to the set of substitutions and not to the term ``substitutive").
In what follows we will only consider sets of substitutions that are endomorphisms. 
As already said, Berthé and Delecroix \cite{Berthe_Delecroix2014RIMS} mentioned that ``\textit{an $S$-adic system is a system that can be indefinitely desubstituted}". The next result formalizes this in our context.

\begin{proposition}
\label{S-adic=>stable}
Let $\cS \subseteq \subst(A)$. Any $\cS$-adic word belongs to $\stab(\cS)$.

Moreover{,} if $(\sigma_n)_{n \geq 1}$ is the directive word  of an $\cS$-adic word $\bw${,} 
then $(\sigma_n)_{n \geq 1}$ is also a directive sequence of $\bw$ as an element of $\stab(\cS)$.
\end{proposition}

Let us observe that Proposition~\ref{S-adic=>stable} is not immediate. 
If $(\sigma_n)_{n \geq 1}$  is the directive sequence of an {$\cS$-adic} word $\bw$, there exists a sequence of letters $(a_n)_{n \geq 1}$ 
such that $\bw = \lim_{n \to \infty} \sigma_1 \sigma_2 \cdots \sigma_n( a_{n})$. 
This does not imply that for all $k \geq 0$, the limit  $\lim_{n \to \infty} \sigma_k \sigma_{k+1} \cdots \sigma_n( a_{n})$ exists. 
Hence the fact that $\bw$ belongs to $\stab(\cS)$ is not immediate. 
This phenomenon can be illustrated by the following example. 
Let $f: a \mapsto a, b \mapsto a$ and
$g: a \mapsto bb, b \mapsto aa$.
For all $n \geq 0$, we have
$g^{2n}(a) = a^{2^{2n}}$ $g^{2n+1}(a) = b^{2^{2n+1}}$.
Hence taking $a_n =a$ for all $n \geq 1$, we have $\lim_{n \to \infty} fg^n(a) = a^\omega$ but $\lim g^n(a)$ does not exist.
Nevertheless taking $\bw_{2n} = a^\omega$ and $\bw_{2n+1} = b^\omega$ provides a sequence of infinite desubstituted words of $a^\omega$ using $f$ and $g$.

{Another difficulty to cope with is that, for some sets $\cS$ of substitutions, there may exist a word $\bw$ that is $\cS$-adic but for which, given any sequence of desubstituted words $(\bw_n)_{n \geq 0}$ of $\bw$, we do not have
$\bw = \lim_{n \to \infty} \sigma_1 \cdots \sigma_n( \first( \bw_n))$. For instance, let $\cS = \{L_a\}$. 
There is only one word in $\stab(\cS)$: the word $a^\omega$. The unique sequence of desubstituted words is $(a^\omega)_{n \geq 0}$ and $a^\omega \neq \lim_{n \to \infty} L_a^n(\first(a^\omega))$. Nevertheless $a^\omega$ is $\{L_a\}$-adic since $a^\omega = \lim_{n \to \infty} a^nb = \lim_{n \to \infty} L_a^n(b)$.}

{The situation in the previous example can be greatly explained by the fact 
that the finite word $a$ can be indefinitely desubstituted. 
In Section~{\protect\ref{subsec:desubstitutable_finite}}, 
we study the set of finite words that can be indefinitely desubstitutable.
In particular, Proposition~{\ref{P:GenStabFin_fini}} states that, 
given a directive sequence $\bs$, 
the set of finite words that can be indefinitely desubstituted with $\bs$
is finitely generated.
The next example shows that this does not mean that 
the set of finite words that can be indefinitely desubstituted using a given set of substitutions
is finitely generated. 
Set $\cS = \{L_a, L_b\}$. 
The words $a$ and $b$ can be indefinitely desubstituted using, 
respectively, $L_a^\omega$ and $L_b^\omega$ as directive sequences.
But not all words over $\{a, b\}$ belong to $\stab(\cS)$.}

{After proving Proposition~{\protect\ref{P:GenStabFin_fini}} 
in Section~{\protect\ref{subsec:desubstitutable_finite}},
we use this proposition to prove Proposition~{\protect\ref{S-adic=>stable}} (in Section~{\protect{\ref{subsec:proof4.1}}}).
In Section~{\protect\ref{subsec:more_on_structure}}, we provide more information on the structure of stable sets.
In Section~{\protect\ref{subsec:s-adic_balanced}}, we consider a question relative to a converse to 
Proposition~{\protect\ref{S-adic=>stable}}.}

\subsection{\label{subsec:desubstitutable_finite}\texorpdfstring{{On desubstitutable finite words}}{}}

{For $\bs \in \cS^\omega$, 
let $\stabfin(\bs)$ and $\stablet(\bs)$
denote the sets of, respectively, finite nonempty words and letters
that can be indefinitely desubstituted using $\bs$ as directive sequence.
Set $\genstabfin(\bs) = \stabfin(\bs)\setminus (\stabfin(\bs)\stabfin(\bs)^+)$: 
this set is the set of all elements of $\stabfin(\bs)$ that cannot be decomposed 
as a concatenation of two or more elements of $\stabfin(\bs)$.
Finally, let $\stabultlet(\bs)$ be the subset of $\stabfin(\bs)$ containing words 
that can be desusbtituted in such a way that desubstituted words are ultimately letters:}
$$\stabultlet(\bs) = \{ f(\alpha) \mid \bs = f \bs', f \in \cS^*, \bs' \in \cS^\omega  {\rm ~and~} \alpha \in \stablet(\bs')\}$$

\begin{proposition}
\label{P:structure_parts_1_2}
Let $\cS \subseteq \subst(A)$ and let $\bs \in \cS^\omega$.
\begin{enumerate}
\item $\stabfin(\bs) = \genstabfin(\bs)^+$
\item {$\genstabfin(\bs) \subseteq \stabultlet(\bs)$}
\end{enumerate}
\end{proposition}

\begin{proof}
Relation~1 follows directly the definition of $\stabfin(\bs)$ and $\genstabfin(\bs)$.

Relation~2 is also a consequence of the definition of the sets $\genstabfin(\bs)$ and {$\stabultlet(\bs)$}.
Let $w \in \genstabfin(\bs)$.
Let $\bs = (\sigma_n)_{n \geq 1}$ and let $(w_n)_{n \geq 0}$ be the corresponding sequence 
of finite desubstituted words: $w_0 = w$
and $w_{n-1} = \sigma_n(w_n)$ for $n \geq 1$.
Observe {that} $|w_{n-1}| \geq |w_n| \geq 1$.
Hence the sequence $(|w_n|)_{n \geq 0}$ is ultimately constant.
If this constant is 2 or more,
then we find a contradiction with {indecomposability} of elements of  $\genstabfin(\bs)$. 
Hence $(|w_n|)_{n \geq 0}$ is ultimately 1 and Relation~2 holds.
\end{proof}

Observe that the inverse inclusion of the second assertion of Proposition~\ref{P:structure_parts_1_2}
does not hold in general.
For instance if $\cS = \{f, \id\}$ with $\id$ the identity morphism and $f$ defined by 
$f(a) = bc$, $f(b) = b$ and $f(c) = c${, then} $bc \not\in \genstabfin(f.\id^\omega)$ 
(with $f.\id^\omega$ the sequence of substitutions beginning with $f$ and followed by $\id^\omega$),
even if $bc = f(a)$, as $b$ and $c$ belong to $\genstabfin(f.\id^\omega)$.
{The same phenomenon exists if one consider, instead of $\id$, any morphism $g$ such that $g(\{a, b, c\})= \{a, b, c\}$.}

\begin{proposition}
\label{P:GenStabFin_fini}
Let $\cS \subseteq \subst(A)$ and {$\bs \in \cS^\omega$}.
The sets $\genstabfin(\bs)$ and {$\stabultlet(\bs)$ are} finite.
{Their cardinalities are bounded by the cardinality of $A$.}

More precisely there exist $p \in \cS^*$ and $\bs' \in \cS^\omega$ such that $\bs = p\bs'$, and, 
for all elements $u$ in $\stabultlet(\bs)$, 
$u = p(\alpha)$ for some $\alpha \in \stablet(\bs')$.
\end{proposition}

\begin{proof}
{By Relation~2 of Proposition~{\protect\ref{P:structure_parts_1_2}}, 
it is sufficient to prove the result for  $\stabultlet(\bs)$.
Let $x_1$, \ldots, $x_k$ be distinct elements of $\stabultlet(\bs)$.}
{T}here exist prefixes $p_1$, \ldots, $p_k$ of $\bs$ 
(each $p_i$ is a composition of morphisms),
suffixes $\bs_1$, \ldots, $\bs_k$ of $\bs$
and letters $\alpha_1$, \ldots, $\alpha_k$ such that,
for all $i$, $1 \leq i \leq k$,
$x_i = p_i(\alpha_i)$, $\alpha_i \in \stablet(\bs_i)$
and $\bs = p_i \bs_i$.
Let $p$ be the longest word among the words $p_1$, \ldots, $p_k$.
Let $\bs'$ be the corresponding suffix of $\bs$: $\bs = p \bs'$.
For each $i$, let $f_i$ such that $p = p_i f_i$.
As $\alpha_i \in \stablet(\bs_i)$,
there exists a letter $\alpha_i' \in \stablet(\bs')$
such that $\alpha_i = f_i(\alpha_i')$.
Hence for all $i$, $x_i = p(\alpha_i')$.
As $\stablet(\bs')$ {is a subset of $A$},
we conclude that {$k \leq \#A$. So $\stabultlet(\bs)$ is finite.}
\end{proof}

\subsection{\label{subsec:proof4.1}\texorpdfstring{{Proof of Proposition \protect\ref{S-adic=>stable}}}{}}

Let $\bw$ be an $\cS$-adic word. 
There exist a sequence $(\sigma_n)_{n \geq 1}$ of elements of $\cS^\omega$ 
and a sequence of letters $(a_n)_{n \geq 1}$ 
such that $\bw = \lim_{n \to \infty} \sigma_1 \sigma_2 \cdots \sigma_n( a_n)$.
{We prove that $\bw \in \stab(\cS)$ 
and that $(\sigma_n)_{n \geq 1}$ is a directive sequence of $\bw$. 
We have to construct a sequence of desubstituted words.}

{Note that, for all $m \geq 1$,  $\lim_{n \to \infty} |\sigma_m \cdots \sigma_n( a_n)| = \infty$.
Let $k \geq 1$ be an integer.
We first consider links between prefixes of length $k$ of words $\sigma_m \cdots \sigma_n( a_n)$ 
(when $k \leq |\sigma_m \cdots \sigma_n( a_n)|$).
Informally, the idea of the proof is that if we can find a $k$ and a sequence $(p_m)_{m \geq 0}$ of such prefixes
such that, for all $m$, $p_m$ is a prefix of $\sigma_m(p_{m+1})$  and $\lim_{m \to \infty} |\sigma_1 \sigma_2 \cdots \sigma_m( p_m)| = \infty$,
then the words $\lim_{n \to \infty} \sigma_m\cdots \sigma_n(p_n)$ would form a sequence of desubstituted words. If we cannot find any $k$ and any sequence of such prefixes, then $\bw$ has infinitely many prefixes that can be indefinitely desubstituted using the directive sequence $(\sigma_n)_{n \geq 1}$ and the construction of desubstituted words would differ.}

{Let ${\cal V}_k$ be the set of all pairs $(p, m)$ with $m \geq 1$ an integer and $p$ a word  of length $k$ 
such that there exists an integer $n \geq m$ for which $p$ is a prefix of $\sigma_m \cdots \sigma_n( a_n)$. 
Let ${\cal G}_k$ be the directed acyclic graph
whose vertices are the elements of ${\cal V}_k$
and the edges are the pairs $((p, m), (p', m+1))$ of elements of ${\cal V}_k$ such that $p$ is the prefix of length $k$ of $\sigma_m(p')$.}

{Let $(p', m+1)$ be an element of ${\cal V}_k$ with $m \geq 1$.
Let $n \geq m+1$ be an integer such that $p'$ is a prefix of $\sigma_{m+1}\cdots\sigma_n(a_n)$.
Let $p$ be the prefix of length $k$ of $\sigma_m(p')$: $p$ is also a prefix of
$\sigma_m\sigma_{m+1}\cdots\sigma_n(a_n)$. Hence $(p,m) \in {\cal V}_k$ and  $((p, m), (p', m+1))$ is an edge of ${\cal G}_k$. 
More precisely,
since $\sigma_m(p')$ has a unique prefix of length $k$,
 $((p, m), (p', m+1))$ is the unique edge of ${\cal G}_k$ ending on $(p', m+1)$. It follows that ${\cal G}_k$ is an infinite forest.}

{Let $\pi_k$ be the prefix of length $k$ of $\bw$.
As $\bw = \lim_{n \to \infty} \sigma_1 \sigma_2 \cdots \sigma_n( a_n)$, 
there exist infinitely many elements $(p,m)$ in ${\cal V}_k$ for which there exists a path from $(\pi_k, 1)$ to $(p, m)$.
Let ${\cal T}_k$ be the subgraph of ${\cal G}_k$ induced by elements of ${\cal V}_k$ that are accessible from $(\pi_k, 1)$.
The graph ${\cal T}_k$ is an infinite tree.
By Konig's lemma, there exists an infinite path $(p_i, i)_{i \geq 1}$ starting from $(\pi_k, 1)$.}

{Let $m, n$ with $n \geq m \geq 1$. 
By construction, $p_m$ is the prefix of length $k$ of $\sigma_m\cdots \sigma_{n-1}(p_n)$
and, by definition of ${\cal V}_k$, $p_n$ is a prefix of 
$\sigma_n\cdots \sigma_\ell(a_\ell)$ for some $\ell \geq 1$.
Hence $p_m$ is the prefix of length $k$ of $\sigma_m\cdots \sigma_\ell(a_\ell)$ for arbitrary large $\ell$.
It follows that $\sigma_1\cdots\sigma_{m-1}(p_m)$ is a prefix of $\sigma_1\cdots \sigma_\ell(a_\ell)$ for arbitrary large $\ell$.
As $\bw = \lim_{n \to \infty} \sigma_1 \sigma_2 \cdots \sigma_n( a_n)$, 
$\sigma_1\cdots\sigma_{m-1}(p_m)$ is a prefix of $\bw$.}

{Assume that $\lim_{n \to \infty} |\sigma_1\cdots\sigma_{n-1}(p_n)| = \infty$.
We have $\bw = \lim_{n \to \infty} \sigma_1\cdots\sigma_{n-1}(p_n)$.
Let $m \geq 1$.
We also have $\lim_{n \to \infty} |\sigma_m\cdots\sigma_{n-1}(p_n)| = \infty$.
Moreover, by construction, for all $n \geq m$,
$\sigma_m\cdots\sigma_{n-1}(p_n)$ is a prefix of $\sigma_m\cdots\sigma_{n}(p_{n+1})$.
Thus $\lim_{n \to \infty} \sigma_m\cdots\sigma_{n-1}(p_n)$ defines an infinite word.
We denote it $\bw_m$.
It follows from the construction that, for all $m \geq 1$,
$\bw_m = \sigma_m(\bw_{m+1})$
and $\bw_1 = \bw$.
Hence $\bw \in \stab(\cS)$ and $\bw$ has $(\sigma_n)_{n \geq 1}$ as directive sequences.}

{To end the proof of Proposition~{\protect\ref{S-adic=>stable}},
we consider the case where for all $k \geq 1$, 
and for all infinite paths $(p_i, i)_{i \geq 1}$ in ${\cal T}_k$, 
$\lim_{n \to \infty} |\sigma_1\cdots\sigma_{n-1}(p_n)|$ is finite.}

{Let $k \geq 1$ be an integer and let $(p_i, i)_{i \geq 1}$ be an infinite path in ${\cal T}_k$.
As, for all $i \geq 1$, $p_i$ is a prefix of $\sigma_i(p_{i+1})$, 
$\sigma_1\cdots\sigma_{i-1}(p_i)$is a prefix of $\sigma_1\cdots\sigma_i(p_{i+1})$.
The fact that $\lim_{n \to \infty} |\sigma_1\cdots\sigma_{n-1}(p_n)|$ is finite
implies that
there exists an integer $N$
such that, for all $n \geq N$, 
$|\sigma_1\cdots\sigma_{n-1}(p_n)| = |\sigma_1\cdots\sigma_{n}(p_{n+1})|$.
It follows that, for all $n \geq N$, $|p_n| = |\sigma_n(p_{n+1})|$
and so
$p_n = \sigma_n(p_{n+1})$.
As for all $n \geq 0$, $\sigma_1\cdots\sigma_{n-1}(p_n)$ is a prefix of 
$\sigma_1\cdots\sigma_{n}(p_{n+1})$ and a prefix of $\bw$, 
the limit $\lim_{n \to \infty} \sigma_1\cdots\sigma_{n-1}(p_n)$
converges to a prefix $\pi$ of $\bw$.}

{For any $n \geq 1$, $|p_n| = k$.
Let us decompose $p_n$ over letters.
For $i$ with $1 \leq i \leq k$, let $p_{n, i} \in A$ such that
$p_n = p_{n,1}\cdots p_{n, k}$.
We have $\pi = [\sigma_1\cdots\sigma_{N-1}(p_{N,1})] \cdots [\sigma_1\cdots\sigma_{N-1}(p_{N,k})]$
and for all $n \geq N$ and for all $i$, $1 \leq i \leq k$,
$p_{n,i} = \sigma_{n}(p_{n+1,i})$.
In other words, the prefix $\pi$ of $\bw$ can be decomposed into $k$ factors  
that are indefinitely desubstitutable with $(\sigma_n)_{n \geq 1}$ as directive sequence:
$\pi \in \stabultlet(\bs)$ with $\bs = (\sigma_n)_{n \geq 1}$.
}

{What is described before holds for arbitrary $k \geq 1$ since we have assumed: for all $k \geq 1$, 
and for all infinite paths $(p_i, i)_{i \geq 1}$ in ${\cal T}_k$, 
$\lim_{n \to \infty} |\sigma_1\cdots\sigma_{n-1}(p_n)|$ is finite.
This means that $\bw$ has infinitely many prefixes in $(\stabultlet(\bs))^*$.
By Proposition~{\protect\ref{P:GenStabFin_fini}}, $\stabultlet(\bs)$ is finite.
Thus $\bw$ belongs to $(\stabultlet(\bs))^\omega$.}

{Let us decompose $\bw$ over $\stabultlet(\bs)$.
Let $(u_i)_{i \geq 1}$ in $\stabultlet(\bs)$ 
such that $\bw = \prod_{i \geq 1} u_{i}$.
For each $i \geq 1$, let $(u_{i, j})_{j \geq 0}$ be the sequence of desubstituted words of $u_i$: $u_i = u_{i, 0}$
and $u_{i, j} = \sigma_{j+1}(u_{i, j+1})$ for $j \geq 0$.
Let $\bw_j = \prod_{i \geq 1} u_{i, j}$.
By construction, $\bw = \bw_0$ and, for $j \geq 0$,
$\bw_j = \sigma_{j+1}(\bw_{j+1})$. So $\bw$ belongs to $\stab(\cS)$ and $(\sigma_n)_{n \geq 0}$ 
is a directive sequence of $\bw$.}

\subsection{\label{subsec:more_on_structure}\texorpdfstring{{More on the structure on stable sets}}{}}

{For $\bs \in \cS^\omega$, 
let $\stab(\bs)$ 
denote the set of infinite words
that can be indefinitely desubstituted using $\bs$ as directive sequence.
Let $\adic(\bs)$ be the set of infinite words that are $S$-adic with $\bs$ as directive sequence.
}

\begin{proposition}
\label{P:structure_part_3}
Let $\cS \subseteq \subst(A)$ and let $\bs \in \cS^\omega$. We have:
$$\stab(\bs) = \genstabfin(\bs)^\omega \cup (\genstabfin(\bs))^*\adic(\bs)$$
\end{proposition}

\begin{proof}

{T}he inclusion
$\genstabfin(\bs)^\omega \cup (\genstabfin(\bs))^*\adic(\bs) \subseteq \stab(\cS)$
follows directly the definitions.
For the converse, assume that $\bw \in \stab(\bs)$. 
Let $\bs = (\sigma_n)_{n \geq 1}$ and let {$(\bw_n)_{n \geq 0}$} be the corresponding sequence of desubstituted words.
Let also $(a_n)_{n \geq 0}$ be the sequence $(\first(\bw_n))_{n \geq 0}$.
If there exist arbitrarily large integers $n$ such that $|\sigma_n(a_n)|  \geq 2$,
then $\bw (= \bw_0)$ belongs to $\adic(\bs)$.
Otherwise we ultimately have $a_{n-1} = \sigma_n(a_n)$.
Set $u= \lim_{n \to \infty} \sigma_1\cdots \sigma_n(a_n)$.
This word $u $ belongs to $\stabfin(\bs)$.
Thus, by Relation~1 of {Proposition~\protect\ref{P:structure_parts_1_2}}, $\bw = u \bw'$ with $u \in \genstabfin(\bs)^+$ and $\bw' \in \stab(\bs)$.
Hence the proof can end by induction on the length of prefixes of $\bw$.
\end{proof}

{Let us observe that $\genstabfin(\bs)^\omega \cup (\genstabfin(\bs))^*\adic(\bs)$ may not be empty. This is the case for instance when $\bs = L_a^\omega$ since $a^\omega = \lim_{n \to \infty} L_a^n(b)$
and $\genstabfin(\bs) = \{a\}$.}

The next result is a direct consequence of Propositions~\ref{P:structure_part_3} and \ref{P:GenStabFin_fini}.

\begin{corollary}
\label{C:structure2}
Let $\cS \subseteq \subst(A)$. An infinite word $\bw$ belongs to $\stab(\cS)$ if and only if one of the two following cases holds:
\begin{itemize}
\item $\bw = f(u \bw')$ with $f \in \cS^*$ and there exists a sequence $\bs \in \cS^\omega$
such that $u \in \stablet(\bs)^*$ and $\bw' \in \adic(\bs)$.
\item $\bw = f(\bw')$ with $f \in \cS^*$ and $\bw' \in (\stablet(\bs))^\omega$ 
for a sequence $\bs \in \cS^\omega$.
\end{itemize}
\end{corollary}

The next result is a direct consequence of Corollary~\ref{C:structure2}.

\begin{corollary}
\label{C:condition_S_adic}
Let $\cS \subseteq \subst(A)$.
{If $\bigcup_{s \in \cS^\omega} \stablet(\bs)$ is empty,
then the set $\stab(\cS)$ is exactly the set of $\cS$-adic words.}
\end{corollary}

{Let us observe that the converse of this lemma does not hold, as shown by
Proposition~{\protect{\ref{S-adicity balanced words}}}.}

Let us observe {also} that the condition in Corollary~\ref{C:condition_S_adic} is decidable when $\cS$ is finite.
Indeed the set $\stablet(\cS)$ can be algorithmically determined as follows. 
First construct the graph whose vertices are letters and whose labeled oriented edges are the triples $(\alpha, f, \beta)$ with $\alpha$, $\beta$ letters and $f \in \cS$ such that $\alpha = f(\beta)$.
Elements of $\stablet(\cS)$ are the vertices from which go at least one infinite path in the graph. 
Labels of such paths are directive sequences of the desubstitutions.
Finally{,}  $\bigcup_{s \in \cS^\omega} \stablet(\bs)$ is empty if and only if there is no circuit in the graph.
{An open question is to decide, given a finite set $\cS$ of substitutions, whether the set $\stab(\cS)$ is exactly the set of $\cS$-adic words.}

\subsection{\label{subsec:s-adic_balanced}Substitutive-adicity of balanced words}

After Corollary~\ref{C:condition_S_adic},
a natural question is:
given a set $\cS \subseteq \subst(A)$, does there exist a set $\cS' \subseteq \subst(A)$
such that $\stab(\cS)$ is the set of all $\cS'$-adic words.
Although Proposition~\ref{S-adicity balanced words} below provides an example of positive answer, 
this does not hold in general as shown by the set $\{\id\}$ with $\id$ the identity morphism over $A^*$:
$\stab(\{\id\}) = A^\omega$ but this set does not correspond to a set of $\cS'$-adic words
as shown by the next lemma. 
{We say that a morphism is a \textit{permutation} if $\{ f(a) \mid a \in A\} = A$.}

\begin{lemma}
\label{L:Aomega_non_adic}
Let $A$ be an alphabet containing at least two letters.
There exists no set $\cS \subseteq \subst(A)$ 
such that $A^\omega$ is exactly the set of all {$\cS$}-adic words.
Moreover any set $\cS$ such that $A^\omega = \stab(\cS)$ must contain
{a permutation.}
\end{lemma}

\begin{proof}
Let $A$ be an alphabet with $k \geq 2$ distinct letters $a_1$, \ldots, $a_k$.
Let $\cS \subseteq \subst(A)$  {be} such that $A^\omega$ is the set of 
all $\cS$-adic words. By Proposition~\ref{S-adic=>stable}, $A^\omega = \stab(\cS)$.
Let $\bw$ be an infinite word containing all finite words {over $A$} as factors.
There exists a substitution $f$ in $\cS$ and an infinite word $\bw'$ such that $\bw = f(\bw')$.
For each $i$, $1 \leq i \leq k$, $\bw$ has arbitrary large factors 
that are powers of $a_i$.
This implies that there exists a letter $\alpha_i \in A$ and an integer $\ell_i$ with $f(\alpha_i)= a_i^{\ell_i}$.
{Note that $A = \{\alpha_i \mid 1 \leq i \leq k\}$. Hence $\bw = f(\bw')$ belongs to $\{ a_1^{\ell_1}, \ldots, a_k^{\ell_k}\}^\omega$.
But} there also exist arbitrarily large factors of $\bw$ {in} $(a_1a_2\cdots a_k)^*${. As $k \geq 2$,} for all $i$, $\ell_i = 1$. 
The morphism $f$ {is a permutation}.

Observe that $\bw'$ contains all finite words as factors. 
Thus, by induction, we can see that any directive sequence of $\bw$ as an $\cS$-adic word
is a sequence of {permutations.}
{This contradicts the fact that, for any directive sequence $(\sigma_n)_{n\geq 0}$ of $\bw$ as an $\cS$-adic word,
it should exist a sequence of letters $(a_n)_{n \geq 0}$ such that $\bw = \lim_{n \to \infty} \sigma_1\cdots\sigma_n(a_n)$.}
\end{proof}

Let us recall that J.~Cassaigne provides an example showing that all words of $A^\omega$ 
are $\cS$-adic with $\cS$ a finite set of substitutions but the substitutions used in this example are defined on a larger alphabet than $A$ 
and the stable set $\stab(\cS)$ contains elements that are not in $A^\omega$ 
{(see, \textit{e.g.}, \protect{\cite[Remark 3]{Berthe_Delecroix2014RIMS}})}.
Let us recall this example. 
Let {$\bw = (a_i)_{i \geq 1}$} in $A^\omega$ ($a_i \in A$). 
We have
$\bw   = \lim_{n \to \infty} f \sigma_2\cdots \sigma_n( \ell )$ 
with $\ell$ a letter that does not belong to $A$ and the morphisms
$f$ and $(\sigma_i)_{i \geq 2}$ defined as follows.

$f: 
\left\{
\begin{array}{l}
\ell \mapsto a_1\\
\alpha \mapsto \alpha ~~~~~~ (\alpha \in A)
\end{array}
\right.
$~~~~~~
$\sigma_i: 
\left\{
\begin{array}{l}
\ell \mapsto \ell a_i\\
\alpha \mapsto \alpha ~~~~~~ (\alpha \in A)
\end{array}
\right.
$

Note also that the previous lemma is not true for a one-letter alphabet as $a^\omega$ is the fixed point of the morphism $d$ defined by $d(a) =aa$.

The next result characterizes balanced words in term{s} of $S$-adicity.

\begin{proposition}
\label{S-adicity balanced words}
An infinite word over $\{a, b\}$ is balanced if and only if
{it is $\cSbal$-adic}.
\end{proposition}

\begin{proof}
{
First assume that a word $\bw$ is $\cSbal$-adic.
By Proposition~{\protect\ref{S-adic=>stable}}, $\bw$ belongs to $\stab(\cSbal)$.
Then by Proposition~{\protect\ref{P:characterising balanced words}}, it is balanced.}

From now on, 
assume that $\bw$ is balanced.
By Proposition~\ref{P:characterising balanced words}, $\bw \in \stab(\cSbal)$.
{Let $\bs = (\sigma_n)_{n \geq 1} \in \cS^\omega$ be a directive sequence of $\bw$.
As already mentioned in the proof of Proposition~{\protect\ref{P:characterising balanced words}},
if $\bs$ contains infinitely many elements in $\{L_a, R_a\}$ and infinitely many elements in $\{L_b, R_b\}$,
then $\bw$ is $\cSbal$-adic.
In the remaining cases, there exist an integer $N \geq 1$ 
and a letter $\alpha \in \{a, b\}$ such that,
for all $n \geq N$, $\sigma_n \in \{L_\alpha, R_\alpha\}$.
Let $f = \sigma_1 \cdots \sigma_{N-1}$ and let $\bw'$ be the element
of $\stab(\cS)$ directed by $(\sigma_n)_{n \geq N}$ such that $\bw = f(\bw')$.
Let $\beta$ be the letter such that $\{\alpha, \beta\} = \{a, b\}$.
Observe that $\bw'$ contains at most one occurrence of $\beta$.
Indeed if it contains a factor $\beta \alpha^n \beta$ for some integer $n$,
then for any $i$, $1 \leq i \leq n$, the $i$th desubstituted word of $\bw'$ (remember that $\sigma_i \in \{L_\alpha, R_\alpha\}$)
would contain the factor $\beta \alpha^{n-i}\beta$.
But the $n$th desubstituted word cannot contain the factor $\beta\beta$
since it can itself be desubstituted using $L_\alpha$ or $R_\alpha$.
Hence $\bw' = \alpha^\omega = \lim_{n \to \infty} L_\alpha^n(\beta)$ or
$\bw' = \alpha^k\beta\alpha^\omega = L_\alpha^k(\lim_{n \to \infty} R_\alpha^n(\beta))$ for some $k \geq 0$.
Hence $\bw'$ and $\bw=f(\bw')$ are $\cSbal$-adic.}
\end{proof}

\section{\label{sec:other_binary_examples}Two other examples of stable sets on binary alphabets}

After Remark~\ref{Rem:preserve}, while searching for examples of stable sets defined by a combinatorial property,
it is natural to consider properties for which we know the morphisms that preserve them.
But this does not provide systematically a stable set as shown by the next example.

An overlap-free word is a word that contains no factor 
{of} the form $\alpha u \alpha u \alpha$ with $\alpha$ a letter
and $u$ a word. 
Since Thue \cite{Thue1912}{,} it is known that the morphisms 
that preserve overlap-free words are the morphisms obtained by compositions of the {morphism} $\mu$ ($\mu(a) = ab$, $\mu(b) = ba$) and the exchange  morphi{s}ms $E$ ($E(a) = b$, $E(b) = a$). As $\mu E = E \mu$, one can check that $\stab(\{\mu\}) = \stab(\{\mu, E\}^* \setminus \{E, \id\}) = \{ \bM, E(\bM)\}$ where $\bM$, the Thue-Morse word, is the fixed point of $\mu$ starting with letter $a$.
It follows that no stable set corresponds to the set of binary overlap-free words.
Also $\bM$ and $E(\bM)$ are the unique overlap-free words 
that can be defined using $S$-adicity for a set $S$ of morphisms preserving overlap-freeness
(as by Proposition~\ref{S-adic=>stable}, $S$-adic words belong to $\stab(\cS)$).

\subsection{\label{subsec:sturm}Sturmian words}

As mentioned previously the Sturmian words are the binary balanced aperiodic words{.}
We show below that the set of Sturmian words can be defined as a stable set but only with infinite sets of substitutions.

In the proof of Proposition~\ref{P:characterising balanced words}
we have already recalled that Sturmian words are {(exactly)} the elements of $\stab(\cSbal)$ such that any directive sequence
contains infinitely many elements of $\{L_a, R_a\}$ and infinitely many elements of $\{L_b, R_b\}$ \cite{BertheHoltonZamboni2006}.
Using the contraction method already used in the proof of Proposition~\ref{P:characterising balanced words}, 
this can be reformulated using the set $\cSsturm = \{L_a, R_a\}^+\{L_b, R_b\} \cup \{L_b, R_b\}^+\{L_a, R_a\}$.
{Thus we have proved the following proposition.}

\begin{proposition}
A word is Sturmian if and only if it belongs to $\stab(\cSsturm)$.
\end{proposition}

\begin{proposition}
\label{No finite set for Sturmian}
There exists no finite set $\cS$ of substitutions such that {the set of Sturmian words is equal to} $\stab(\cS)$.
\end{proposition}

For the proof of this proposition we need the next lemma.

\begin{lemma}[\mbox{see, \textit{e.g.}, \cite[chap. 2]{Lothaire2002}}]
\label{L:Sturmian morphisms}
Morphisms that preserve Sturmian words are the elements of the set $(\cSbal \cup \{E\})^*$.
\end{lemma}

 \begin{proof}[Proof of Proposition~\ref{No finite set for Sturmian}]
Assume that the set of Sturmian words is the set $\stab(\cS)$ for some finite set $\cS$ of substitutions.
First observe that, by Remark~\ref{Rem:preserve}, for all $f$ in $\cS$, 
$f$ preserves the family of Sturmian words. 
Hence by Lemma~\ref{L:Sturmian morphisms}, $\cS \subseteq \{L_a, L_b, R_a, R_b, E\}^*$.
As $E L_a = L_b E$, $E L_b = L_a E$, $E R_a = R_b E$, $E R_b = R_a E$ and $EE = \id$,
{we have} $\cS \subseteq \{L_a, L_b, R_a, R_b\}^* \{ \id, E\}$.
Hence there exist two subsets ${\cal F}$  and ${\cal G}$ of $\cSbal^*$
such that $\cS = {\cal F}\cup {\cal G}E$ .
As $\cS$ is finite, let ${\cal F} = \{ f_i \mid 1 \leq i \leq \ell\}$ and let ${\cal G} = \{ g_i \mid 1 \leq i \leq k\}$.

{The following fact is important.}

{
\textbf{Fact.} For $i \geq 0$ an integer, for $f \in \cSbal^i$ and for an infinite word $\bw$, 
 if the word $f(\bw)$ contains a factor $ba^{j+2}b$ with $j \geq i$, then $f \in \{L_a, R_a\}^i$.
}

{
Let us prove this fact by induction on $i$. It is basically true for $i = 0$.
Assume $i \geq 1$.
Observe that for any infinite word $\bu$, both words $L_b(\bu)$ and $R_b(\bu)$ does not contain $aa$.
As $aa$ is a factor of $f(\bw)$ and $f \in \cSbal^i$, it follows that
$f = L_a g$ or $f = R_a g$ with $g \in \cSbal^{i-1}$. 
Whatever is the decomposition of $f$, 
$g(\bw)$ contains a factor $ba^{j+1}b = ba^{(j-1)+2}b$.
As $i-1 \leq j-1$, by induction, $g \in \{L_a, R_a\}^{i-1}$.
hence $f \in \{L_a, R_a\}^i$.
}

{Now let $M = 2\max( \{ |f_i| \mid 1 \leq i \leq \ell\} \cup \{ |g_i| \mid 1 \leq i \leq k\})$ 
(here, for $f \in \cSbal^*$, $|f|$ is the length of $f$ considered as a word over the alphabet $\cSbal$).}
{Let $\bw$ be a Sturmian word that contains the factor $ba^{M+2}b$.
Let $(\sigma_n)_{n \geq 1} \in \cS^\omega$ be a directive word of $\bw$ (as an element of $\stab(\cS)$).}

For any morphism $f$ in $(\cSbal \cup E)^*$,
let $\bar{f}$ be the morphism obtained from the decomposition of $f$ over $\{L_a, L_b, R_a, R_b\}$
replacing each occurrence of $L_a$ with $L_b$, each occurrence of $R_a$ with $R_b$, each occurrence of $L_b$ with $L_a$ and each occurrence of $R_b$ with $R_a$.  
Observe that $f E = E \bar{f}$.

{Let $\sigma_1'$, $\sigma_2'$ be the morphisms defined as follows.}
\begin{itemize}
\item if $\sigma_1 \in {\cal F}$ and $\sigma_2 \in  {\cal F}$, set $\sigma_1' = \sigma_1$,  $\sigma_2' = \sigma_2$ (note that $\sigma_1'\sigma_2' = \sigma_1\sigma_2$);
\item if $\sigma_1 \in {\cal F}$ and $\sigma_2 \in  {\cal G}E$, set $\sigma_1' = \sigma_1$,  $\sigma_2' = \sigma_2 E$ (note that $\sigma_1'\sigma_2' = \sigma_1\sigma_2 E$);
\item if $\sigma_1 \in {\cal G}E$ and $\sigma_2 \in  {\cal F}$, set $\sigma_1' = \sigma_1 E$,  $\sigma_2' = \overline{\sigma_2}$ (note that $\sigma_1'\sigma_2' = \sigma_1\sigma_2 E$);
\item if $\sigma_1 \in {\cal G}E$ and $\sigma_2 \in  {\cal G}E$, set $\sigma_1' = \sigma_1 E$,  $\sigma_2' = \overline{\sigma_2} E$ (note that $\sigma_1'\sigma_2' = \sigma_1\sigma_2$).
\end{itemize}

In all cases, {$\sigma_1'$} and {$\sigma_2'$} belong to $\cSbal^*$
and {$\bw = \sigma_1'\sigma_2'(\bw')$} for some infinite word $\bw'$.
{As $|\sigma_1'\sigma_2'|\leq M$ and as $\bw$ contains the factor $ba^{M+2}b$, by the initial fact, 
we} must have $\sigma_1'\sigma_2' \in \{L_a, R_a\}^*$, and so $\sigma_1'$, $\sigma_2' \in \{L_a, R_a\}^*$

If $\sigma_1 \in {\cal F}$, we get $\sigma_1 \in \{L_a, R_a\}^*$.
If $\sigma_1 \not\in {\cal F}$ and $\sigma_2 \in {\cal F}$, {from} $\sigma_2' \in \{L_a, R_a\}^*$, {we get} $\sigma_2 \in \{L_b, R_b\}^*$.
If $\sigma_1 \not\in {\cal F}$ and $\sigma_2 \not\in {\cal F}$, we {have} $\sigma_1\sigma_2 \in \{L_a, R_a\}^*$.
{But then, one of the words $a^\omega$ or $b^\omega$ belongs to $\stab(\cS)$ 
with $\sigma_1^\omega$, $\sigma_2^\omega$ or $(\sigma_1\sigma_2)^\omega$ as  
a directive sequence. As these words are not Sturmian, we have a contradiction with the hypothesis that $\stab(\cS)$ is the set of Sturmian words.}
\end{proof}

The previous result shows that the set of aperiodic words of a stable set of a finite set of substitutions may not form also a stable set of a finite set of substitutions by themselves.

\subsection{Lyndon Sturmian words\label{subsec:Lyndon_Sturmian}}

Let {us} recall that an \textit{infinite Lyndon word} is a word smaller, with respect to the lexicographic order, than its suffixes.
Here{, without loss of generality, we restrict our attention on the ordered alphabet $\{a < b\}$, that is on the alphabet $\{a, b\}$ with $a < b$.
Results for the order $\{ b < a \}$ can be obtained exchanging the roles of $a$ and $b$.}
Let us show that the set of binary balanced infinite Lyndon words form a stable set.
Set $\cS_{Lynd} = \{L_a^nR_b, R_b^nL_a \mid n \geq 1\}$.

\begin{proposition}~
\begin{itemize}
\item A word is a Lyndon Sturmian word over $\{a < b\}$
 if and only if 
it belongs to $\stab(\cS_{Lynd})$
if and only if it is $\cS_{Lynd}$-adic.
\item There is no finite set {$\cS$} such that the set of Lyndon Sturmian word{s} is $\stab(\cS)$.
\end{itemize}
\end{proposition}

\begin{proof}
The first part is a direct consequence of Theorems~5.6 and 6.5 in \cite{LeveRichomme2007TCS} from which we have: \textit{A {Sturmian} word $\bw$ is a Lyndon word over $\{a < b\}$ or over $\{b < a\}$ if and only if there exists a sequence $(d_n)_{n \geq 0}$ of integers such that
$d_1 \geq 0$, $d_k \geq 1$ for all $k \geq 2$ and}
$$\bw = \lim_{n \to \infty} L_a^{d_1}R_b^{d_2}L_a^{d_3}R_b^{d_4} \cdots L_a^{d_{2n-1}}R_b^{d_{2n}}(a)$$
\textit{or}
$$\bw = \lim_{n \to \infty} L_b^{d_1}R_a^{d_2}L_b^{d_3}R_a^{d_4} \cdots L_b^{d_{2n-1}}R_a^{d_{2n}}(a)$$
Consequently 
a Sturmian word $\bw$ is a Lyndon word over $\{a < b\}$ if and only if 
it is {$\cS_{Lynd}$}-adic. 
{Observe that, for any $f$ in $\cS_{Lynd}$, $|f(a)| \geq 2$ and $|f(b)| \geq 2$.
Hence the set $\bigcup_{\bs \in \cS_{Lynd}^\omega} \stablet(\bs)$ is empty.
By Corollary~{\protect\ref{C:condition_S_adic}}, the set of $\cS_{Lynd}$-adic words is equal to $\stab(\cS_{Lynd})$.
}

The second part is a consequence of the fact that morphisms that preserve Lyndon Sturmian words
are the elements of $\{L_a, R_b\}^*$ \cite{Richomme2003BBMS}.
Assume that the set of Lyndon Sturmian words is $\stab(\cS)$
for some $\cS \subseteq \subst(\{a, b\})$.
By Remark~\ref{Rem:preserve}, $\cS \subseteq \{L_a, R_b\}^*$.
But no element of $\cS$ can belong to $L_a^*$ or $R_b^*$ as  $a^\omega$ and $b^\omega$ 
(that are indefinitely desubstitutable over $L_a$ and $R_b$ respectively) are not Lyndon words.
Hence any element $f \in \cS$ should be a concatenation of elements of $\{L_a, R_b\}$
with at least one occurrence of $L_a$ and one occurrence of $R_b$.
But then{,} as in the proof of Proposition~\ref{No finite set for Sturmian},
$\cS$ cannot be finite as otherwise the number of occurrences of the letter $a$ between two occurrences of the letter $b$ would be bounded.
\end{proof}

\section{Episturmian words\label{sec:episturmian}}

In \cite{DroubayJustinPirillo2001,JustinPirillo2002} 
an \textit{episturmian word} is defined as an infinite word whose set of factors is closed under reversal
and that has at most one left special factor of each length 
(Let {us} recall that $u$ is a \textit{left special factor} of a word $\bw$ if $au$ and $bu$ are factors of $\bw$ for $a$ and $b$ two different letters).
The subset of the \textit{strict episturmian words} (or $A$-\textit{strict episturmian words} with $A$ the alphabet) 
corresponds to the \textit{Arnoux-Rauzy words}. 
It is the set of infinite words having exactly one left special factor of each length, 
whose set of factors is closed under {reversal}
and such that if $u$ is a left special factor {then} all words $\alpha u$, for $\alpha \in A$, are also factors.
Also an episturmian word is \textit{standard}
if  its left special factors are its prefixes.

The following morphisms are important when studying episturmian words. They extend {to} arbitrary alphabet
the morphisms of $\cSbal$, and the exchange morphism $E$. For $a, b \in A$, let $E_{a, b}$ be the substitution 
defined by $E_{a, b}(a) = b$, $E_{b,a}(b) = a$ and $E_{b,a}(c) = c$ for $c \in A\setminus\{a, b\}$. 
For $\alpha \in A$, let $L_\alpha$ and $R_\alpha$ be the following substitutions:

$L_\alpha: \left\{
\begin{tabular}{l}
$\alpha \mapsto \alpha$\\
$\beta \mapsto \alpha \beta$ for $\beta \neq \alpha$, $\beta \in A$
\end{tabular}\right.$
\hspace{0.3cm}
$R_\alpha: \left\{
\begin{tabular}{l}
$\alpha \mapsto \alpha$\\
$\beta \mapsto \beta\alpha$ for $\beta \neq \alpha$, $\beta \in A$
\end{tabular} \right.$

Let ${\cal L} = \{ L_\alpha \mid \alpha \in A\}$ and  ${\cal R} = \{ R_\alpha \mid \alpha \in A\}$.
{Let $\exch$ be the set $\{ E_{\alpha, \beta} \mid \alpha, \beta \in A\}$.
Observe that elements of $\exch^*$ are the permutations.
The set $\exch^*$ is finite and its cardinality is $\#A!$.}
{Here follows some useful relations.
For pairwise distinct letters $\alpha$, $\beta$, $\gamma$, 
we have
}
\begin{align}
L_\alpha E_{\beta,\gamma}= E_{\beta,\gamma}L_\alpha \label{eq1}\\
R_\alpha E_{\beta,\gamma}= E_{\beta,\gamma}R_\alpha \label{eq2}\\
L_\alpha E_{\alpha,\beta}= E_{\alpha,\beta}L_\beta  \label{eq3}\\
R_\alpha E_{\alpha,\beta}= E_{\alpha,\beta}R_\beta \label{eq4}
\end{align}

\subsection{Standard episturmian words and stable sets}

As far as we know, the next statement is the unique one that characterized explicitly a combinatorial family of words in term{s} of a stable set before the current paper.

\begin{proposition}[\mbox{\cite[Cor. 2.7]{JustinPirillo2002}}]\label{Y1}\label{P:stable_standard_episturmien}

The set of standard episturmian words is $\stab({\cal L})$.
\end{proposition}

Strict standard episturmian words correspond to the words whose directive sequence contains
infinitely many $L_\alpha$ for each letter $\alpha$ (see \cite{JustinPirillo2002}).
In the binary case they correspond to the standard Sturmian words, also
called characteristic words, that is, to the infinite binary words such that their
prefixes are exactly their special factors.
It is known that the standard Sturmian words are the ${\cal L}$-adic words (see, \textit{e.g.}, \cite[Prop. 2.2.24]{Lothaire2002}).
This does not extend to strict standard episturmian words for alphabet{s} with three letters or more.
Indeed over an alphabet containing three letters or more,
being ${\cal L}$-adic does not imply that there exists a directive sequence such that,
for each letter $\alpha$, 
$L_\alpha$ occurs infinitely often in the directive sequence.
{For instance, the episturmian word directed by $L_c(L_aL_b)^\omega$ is ${\cal L}$-adic but it not a strict standard episturmian word.}

Let ${\cal L}_{StrictStand}$ be the set of all elements of ${\cal L}^*$
whose decompositions on ${\cal L}$ 
{contain} at least one occurrence of each element of ${\cal L}$
and the last substitution in the decomposition has no previous occurrence in the decomposition:
{$f \in {\cal L}_{StrictStand}$} if $f = L_{a_1}L_{a_2}\cdots L_{a_k}$ {for some letters $a_1, \ldots, a_k$, if
for each $\alpha \in A$, there exists $i$,  $1 \leq i \leq k$ such that $a_i = \alpha$, and if $a_k \not\in \{a_1, \ldots, a_{k-1}\}$.
(Observe that the decomposition $L_{a_1}L_{a_2}\cdots L_{a_k}$ is unique:
This can be verified directly using the fact that, for an infinite word $\bw$ containing $a$, $L_a(\bw)$ contains $aa$}) 
For instance, when $A = \{a, b\}$, {${\cal L}_{StrictStand}$} $= \{ L_a^n L_b, L_b^nL_a \mid n \geq 1\}$.
From {the definition of strict standard episturmian words and from the definition of the set  ${\cal L}_{StrictStand}$}, the first part of the next {proposition} may be directly verified.

\begin{proposition}
\label{P:stable_standard_episturmian}
Let $A$ be an alphabet containing at least two letters.
\begin{itemize}
\item The set of $A$-strict standard episturmian words
is the set {$\stab({\cal L}_{StrictStand})$}
\item There exists no finite set $\cS \subseteq \subst(A)$ such that 
{$\stab({\cal L}_{StrictStand}) = \stab(\cS)$}.
\item The set of standard episturmian words is the set of
{${\cal L}$}-adic words.
\end{itemize}
\end{proposition}

\begin{proof}
{The proof of the second part is 
similar to the proof of Proposition~{\protect\ref{No finite set for Sturmian}}
but we have to face a combinatorial problem due to the fact that the alphabet may have more than two letters.
}

{
Assume, by contradiction, that the set of $A$-strict standard episturmian words, 
($\stab({\cal L}_{StrictStand})$ by the first part of the proposition)
is equal to $\stab(\cS)$ for some finite set $\cS$ of substitutions.
First observe that, by Remark~{\protect\ref{Rem:preserve}}, for all $f$ in $\cS$, 
$f$ preserves $A$-strict standard episturmian words. 
Theorem 10 in {\protect\cite{DroubayJustinPirillo2001}} states that \textit{a substitution $\sigma \in \subst(A)$
preserves $A$-strict standard episturmian words
if and only $\sigma \in ({\cal L} \cup \exch)^*$}.
Hence $\cS \subseteq ({\cal L} \cup \exch)^*$.
By Equations~{\protect{\eqref{eq1}}} and {\protect{\eqref{eq3}}},  $\cS^* \subseteq {\cal L}^*\exch^*$ :  any element $\sigma$ of $\cS^*$ can be decomposed into $\sigma = f \pi$ with $f \in {\cal L}^*$ and $\pi$ a permutation.
}

{
Let $a$ and $b$ be two different letters.
The following fact is important.
}

{
\textbf{Fact.} Let $i$, $j$ be integers such that $i \leq j$. 
Let $\bw \in A^\omega$ and $f \in {\cal L}^i$. 
If $f(\bw)$ contains a factor $ba^{j+2}b$,
then $f = L_a^i$.
}

{
We prove this fact by induction. It is basically true for $i = 0$.
Assume $i \geq 1$.
Observe that for $c \neq a$ and for any infinite word $\bu$, $L_c(\bu)$ does not contain $aa$.
As $aa$ is a factor of $f(\bw)$ and as $f \in {\cal L}^i$, it follows that
$f = L_a g$ with $g \in {\cal L}^{i-1}$. 
Consequently, $g(\bw)$ contains a factor $ba^{(j-1)+2}b$ and $i-1 \leq j-1$.
Hence $g = L_a^{i-1}$ by induction. So $f = L_a^i$.
}

{
Now let $M_1 = \max( \{ |f| \mid f\pi \in \cS, f \in {\cal L}^*, \pi \in \exch^*\}$
(here, for $f \in {\cal L}^*$, $|f|$ is the length of $f$ considered as a word over the alphabet ${\cal L}$).
Let
$M = (\#A!+1)M_1$ 
Let $\bw$ be a strict standard episturmian word containing the factor $ba^{M+2}b$.
Let $(\sigma_n)_{n \geq 1} \in \cS^\omega$ be a directive word of $\bw$ (as an element of $\stab(\cS)$).
}

{
Let $i$, $1 \leq i \leq \#A!+1$. 
For $1 \leq j \leq i$, as $\sigma_j \in {\cal L}^*\exch^*$,
there exist $g_j \in {\cal L}^*$ and $\pi_j' \in \exch^*$
such that $\sigma_j = g_j \pi_j'$.
Using Equations~{\protect{\eqref{eq1}}} and {\protect{\eqref{eq3}}}, 
we can see that
there exist $f_i \in {\cal L}^*$ and $\pi_i \in \exch^*$ such that $\sigma_1\cdots\sigma_i = f_i \pi_i$.
Moreover $|f_i| = \sum_{j = 1}^i |g_j|$ and so $|f_i| \leq i M_1 \leq M$.
From the previous fact, $f_i \in\{L_a\}^*$.
Since the cardinality of $\exch^*$ is $\#A!$, there exists $i$ and $j$ such that $1 \leq i < j \leq \#A!+1$, such that $\pi_i = \pi_j$. Set $\pi = \pi_i$.
}

{
Let $\ell_i$ and $\ell_j$ be the integers such that $\sigma_1 \ldots \sigma_i = L_a^{\ell_i} \pi$ 
and $\sigma_1 \ldots \sigma_j = L_a^{\ell_j}\pi$. We have $L_a^{\ell_j}\pi =$ 
$\sigma_1 \ldots \sigma_j =$ 
$L_a^{\ell_i}\pi\sigma_{i+1}\cdots\sigma_j$. Thus 
$\pi\sigma_{i+1}\cdots\sigma_j = L_a^{\ell_j-\ell_i}\pi$.
By Equations~{\protect{\eqref{eq1}}} and {\protect{\eqref{eq3}}}, 
there exists a letter $c$ such that
$L_a^{\ell_j-\ell_i}\pi = \pi L_c^{\ell_j-\ell_i}$. So $\sigma_{i+1}\cdots\sigma_j =$ $L_c^{\ell_j-\ell_i}$.
}

{
Now observe that the word $c^\omega$ can be indefinitely desubstituted using the morphism $L_c^{\ell_j-\ell_i}$.
Thus $c^\omega \in \stab(\cS)$. But $c^\omega$ is not a strict standard episturmian word as it does not contain all the letters of $A$. This contradicts the fact that $\stab(\cS)$ is the set of strict standard episturmian words.
}

\medskip

{Let us prove the third part of Proposition~{\protect\ref{P:stable_standard_episturmian}}.}
By Proposition~\ref{Y1}, standard episturmian words are the elements of $\stab(\cal L)$.
If an element $\bw$ of $ \stab(\cal L)$ has a directive sequence $(\sigma_n)_{n \geq 1}$
containing infinitely many occurrences of at least two elements of ${\cal L}$ 
and  if {$(\bw_n)_{n \geq 0}$} is the corresponding sequence of desubstituted words,
then
one can verify 
that $\lim_{n \to \infty} |\sigma_1\cdots\sigma_n(\first(\bw_n))|= \infty$, 
and so, $\bw$ is {${\cal L}$}-adic.
If it is ultimately $L_\alpha^\omega$, then $\bw = f(\alpha^\omega)$ for some $f \in {\cal L}^*$.
{As there exists a letter $\beta$ different from $\alpha$, 
$\alpha^\omega = \lim_{n \to \infty} L_\alpha^n(\beta)$: $f(\alpha^\omega)$ is ${\cal L}$-adic.}

{Conversely, by Proposition~{\protect\ref{S-adic=>stable}}, the set of $\cal L$-adic words is a subset of $\stab(\cal L)$.}
\end{proof}

Let us recall some results on the family of \textit{LSP words}.
These words were introduced by G.~Fici in 2011 \cite{Fici2011TCS} as the words having all {their} left special factors as prefixes.
Standard episturmian words are particular LSP words.
The set of factors of these words is not necessarily closed by reversal.
In \cite{Richomme2019IJFCS} it was proved that the set of LSP words over $\{a, b\}$ is
the set $\stab(\{L_a, L_b\})$.
By Proposition~\ref{Y1} 
this means that on a binary alphabet, a word is standard episturmian if and only if it is LSP.
This does not hold on larger alphabets. 
Indeed by  \cite{Richomme2019IJFCS}
there exist no finite set $\cS \subseteq \subst(A)$ such that the set
of LSP words over $A$ is $\stab(\cS)$.

\subsection{Recurrence, stable sets and Episturmian words}

Justin and Pirillo \cite{JustinPirillo2002} obtained
also the next result on desubstitutions of episturmian words.
Let {us} recall that a word is \textit{recurrent} if all its factors occur infinitely often, or equivalently,
if all its factors occur at least twice.

\begin{theorem}[\mbox{\cite[Th. 3.10]{JustinPirillo2002}}]
\label{T:charac_episturmien_JP}
An infinite word $\bw$ is episturmian
if and only if 
$\bw \in \stab({\cal L} \cup {\cal R})$ and
it has a sequence {$(\bw_n)_{n \geq 0}$}
of \textbf{recurrent} desubstituted words.
\end{theorem}

{The p}revious theorem and Proposition~\ref{P:characterising balanced words}
show that, in the binary case,
the episturmian words are the recurrent balanced words.
On arbitrary alphabet, this characterization of episturmian words can be simplified {as follows}.

\begin{proposition}
\label{P:another_characterization_of_episturmian_words}
An infinite word $\bw$ is episturmian if and only if 
it is a recurrent element of $\stab({\cal L} \cup {\cal R})$.
\end{proposition}

To prove this {proposition} we need the next two lemmas. The first one is a slight variation of Theorem 5.2 in \cite{GlenLeveRichomme2008TIA} that concerns only the directive sequence and not the sequence of desubstituted words.

\begin{lemma}
\label{L:normalization}
Let $\bw$ be an element of $\stab({\cal L} \cup {\cal R})$.
There exist 
a sequence of infinite words $(\bw_n)_{n \geq 0}$
and a sequence $(\sigma_n)_{n \geq 1}$ in $({\cal L} \cup {\cal R})^\omega$
such that $\bw_0 = \bw$, for all $n \geq 0$, $\bw_n = \sigma_{n+1}(\bw_{n+1})$,
and, for all {$k \geq 0$}, if $\bw_k$ begins with the letter $\alpha$, then $\sigma_{k+1} \neq R_\alpha$.
\end{lemma}

\begin{proof}
As $\bw \in \stab({\cal L} \cup {\cal R})$,
there exist sequences $(\bw_n)_{n \geq 0}$ and $(\sigma_n)_{n \geq 1}$ 
such that $\bw_0 = \bw$ and for all $n \geq 0$, $\bw_n = \sigma_{n+1}(\bw_{n+1})$.

{We construct sequences $(\bw_n')_{n \geq 0}$ and $(\sigma_n')_{n \geq 1}$ 
such that $\bw_0' = \bw$, $\bw_n' =$ $\sigma_{n+1}'(\bw_{n+1}')$ for all $n \geq 0$, 
and, for all $k \geq 0$, if $\bw_k'$ begins with the letter $\alpha$, then $\sigma_{k+1}' \neq R_\alpha$.
It may be emphasize that for all $k \geq 0$, there will exist an integer $n_k \geq 0$, such that $\bw_k = \first(\bw_k)^{n_k} \bw_k'$.
The construction is done by an infinite induction. Set $\bw_0' = \bw_0$.}

{Let $k \geq 0$ be an integer such that $(\bw_i')_{0 \leq i \leq k}$ 
and $(\sigma_i')_{1 \leq i \leq k}$ are already defined.}

{Assume first that $\bw_k' = \bw_k$.}

{Let $\alpha =  \first(\bw_k)$.
If $\sigma_{k+1} \neq R_\alpha$ (that is $\sigma_{k+1} = R_\beta$ with $\beta \neq \alpha$ or  $\sigma_{k+1} = L_\beta$ (possibly $\beta = \alpha$)), 
set $\bw_{k+1}' =$ $\bw_{k+1}$
and $\sigma_{k+1}' =$ $\sigma_{k+1}$.}

{If $\sigma_{k+1} = R_\alpha$, from $\bw_k = R_\alpha(\bw_{k+1})$,
we get $\first(\bw_{k+1}) =$ $\alpha$. Let $\bw_{k+1}'$ be the word such that $\bw_{k+1} = \alpha \bw_{k+1}'$.
Set $\sigma_{k+1}'  = L_\alpha$.
We have $\bw_k' =$ $\bw_k =$
$R_\alpha(\bw_{k+1}) =$ 
$R_\alpha(\alpha \bw_{k+1}') =$ 
$\alpha R_\alpha(\bw_{k+1}') =$ 
$L_\alpha(\bw_{k+1}') =$ 
$\sigma_{k+1}'(\bw_{k+1}')$.}

{To continue the induction step, we now have to consider cases where $\bw_k = \alpha^n \bw_k'$ 
for some letter $\alpha$ and some integer $n \geq 1$.
Observe that, since $\alpha^n \bw_k' = \bw_k = \sigma_{k+1}(\bw_{k+1})$,
we cannot have $\sigma_{k+1} = L_\beta$ with $\beta \neq \alpha$.
Observe also that if $\sigma_{k+1} = L_\alpha$, then $\bw_{k+1}$ begins with $\alpha^{n-1}\beta$ where $\beta = \first(\bw_k')$.}

{If $\sigma_{k+1} = L_\alpha$ and $\alpha \neq \beta$, 
let $\bw_{k+1}'$ be the
word such that $\bw_{k+1} = \alpha^{n-1}\bw_{k+1}'$. 
Set $\sigma_{k+1}' = R_\alpha$. From $\alpha^n\bw_k'=$ $\bw_k =$ $L_\alpha(\bw_{k+1}) =$ 
$L_\alpha(\alpha^{n-1}\bw_{k+1}') =$
$\alpha^{n-1}L_\alpha(\bw_{k+1}')$,
we get $\alpha\bw_k' = L_\alpha(\bw_{k+1}')$ and so $\bw_k' =$ $R_\alpha(\bw_{k+1}) =$ $\sigma_{k+1}'(\bw_{k+1}')$.}

{If $\sigma_{k+1} = L_\alpha$ and $\alpha = \beta$,
let $\bw_{k+1}'$ be the word such that $\bw_{k+1} = \alpha^n \bw_{k+1}'$.
Set $\sigma_{k+1}'  = L_\alpha$.
From $\alpha^n \bw_k' =$ $\bw_k =$ $L_\alpha(\bw_{k+1}) =$ $L_\alpha(\alpha^n\bw_{k+1}') = \alpha^n L_\alpha(\bw_{k+1}')$,
we get $\bw_k' = L_\alpha(\bw_{k+1}')=$ $\sigma_{k+1}'(\bw_{k+1}')$.}

{If $\sigma_{k+1} = R_\beta$ for some letter $\beta\neq \alpha$,
from $\alpha^n\bw_k' = \bw_k = R_\beta(\bw_{k+1})$,
we deduce successively that $\first(\bw_{k+1}) = \alpha$, $\bw_k$ begins with $\alpha\beta$ and $n = 1$.
Let $\bw_{k+1}'$ be the word such that $\bw_{k+1} = \alpha \bw_{k+1}'$.
Set $\sigma_{k+1}'  = L_\beta$.
From $\alpha \bw_k' =$ $R_\beta(\alpha\bw_{k+1}')
=$ $\alpha \beta R_\beta(\bw_{k+1}')$, we get
$\bw_k' =$ $\beta R_\beta(\bw_{k+1}') =$ $L_\beta(\bw_{k+1}')=$ 
$\sigma_{k+1}'(\bw_{k+1}')$.
}

{If $\sigma_{k+1} = R_\alpha$,
from $\alpha^n\bw_k' = \bw_k = R_\alpha(\bw_{k+1})$,
we deduce that $\bw_{k+1}$ begins with $\alpha^n\beta$ with $\beta$ the first letter of $\bw_k'$.
If $\alpha \neq \beta$, 
let $\bw_{k+1}'$ be the word such that $\bw_{k+1} = \alpha^n \bw_{k+1}'$.
Set $\sigma_{k+1}'  = R_\alpha$.
From $\alpha^n \bw_k' =$ $R_\alpha(\alpha^n\bw_{k+1}')
=$ $\alpha^n R_\alpha(\bw_{k+1}')$, we get
$\bw_k' =$ $R_\alpha(\bw_{k+1}') =$ $\sigma_{k+1}'(\bw_{k+1}')$.
If $\bw_{k+1}$ begins with $\alpha^{n+1}$,
let $\bw_{k+1}'$ be the word such that $\bw_{k+1} = \alpha^{n+1} \bw_{k+1}'$.
Set $\sigma_{k+1}'  = L_\alpha$.
From $\alpha^n \bw_k' =$ $R_\alpha(\alpha^{n+1}\bw_{k+1}')
=$ $\alpha^n \alpha R_\alpha(\bw_{k+1}')$, we get
$\bw_k' =$ $\alpha R_\alpha(\bw_{k+1}') =$ 
$L_\alpha(\bw_{k+1}')=$
$\sigma_{k+1}'(\bw_{k+1}')$.
}

{Observe that, in all cases, we have by construction $\sigma_{k+1}' \neq R_{\first(w_{k}')}$.}
\end{proof}

\begin{lemma}
\label{L:recurrence}
Let $\bw$ be an infinite word and let $\alpha$ be a letter.
\begin{itemize}
\item If $L_\alpha(\bw)$ is recurrent then $\bw$  is recurrent.
\item If $R_\beta(\alpha \bw)$ is recurrent then $\alpha \bw$ is recurrent for any letter $\beta \neq \alpha$.
\end{itemize}
\end{lemma}

\begin{proof}
Assume by contradiction that $\bw$ is not recurrent.
It has a factor $u$ that occurs only once.
Let $v = L_\alpha(u)\alpha$. Any occurrence of $v$  in $L_\alpha(\bw)$ must come from an occurrence of $u$ in $\bw$.
This contradicts the recurrence property of $L_\alpha(\bw)$.

Assume by contradiction that $\alpha\bw$ is not recurrent.
{As $\alpha \neq \beta$, we have $R_\beta(\alpha \bw) =$ $\alpha\beta R_\beta(\bw) =$ $\alpha L_\beta(\bw)$. As} $R_\beta(\alpha \bw)$ is recurrent, $L_\beta(\bw)$ is recurrent and{,} from what precedes{,} $\bw$ is also recurrent. Hence $\alpha\bw$ has a prefix $v$ such that $\alpha v$ occurs only once in $\alpha \bw$.
As $R_\beta(\alpha \bw)$ is recurrent, 
the word $R_\beta(\alpha v)$ occurs at a non-prefix position in $R_\beta(\alpha \bw)$.
This implies that $\alpha v$ also occurs at a non-prefix position in $\alpha \bw$. A contradiction.
\end{proof}

\begin{proof}[Proof of Proposition~\ref{P:another_characterization_of_episturmian_words}]
Assume that $\bw$ is a recurrent element of $\stab({\cal L} \cup {\cal R})$.
There exist
sequences $(\bw_n)_{n \geq 0}$ and $(\sigma_n)_{n \geq 1}$ 
as in the statement of Lemma~\ref{L:normalization}.
By induction, using Lemma~\ref{L:recurrence}
we get the recurrence of all words $\bw_n$.
By Theorem~\ref{T:charac_episturmien_JP}, $\bw$ is episturmian.

The converse also follows  immediately from Theorem~\ref{T:charac_episturmien_JP}.
\end{proof}

One can observe that while desubstituting balanced words in the proof of 
Proposition~\ref{P:characterising balanced words},
we get directly {a sequence of desubstituted words and a directive sequence} as in Lemma~\ref{L:normalization}.

\subsection{Characterizing the set of episturmian words as a stable set}

Let ${\cal E}$ be the set of episturmian words.
Proposition~\ref{P:stab_epi} provides characterizations of ${\cal E}$ 
as a stable set and as a set of adic words. 
We first need a technical result.

\begin{lemma}
\label{L:form_direct_episturmian}
Any episturmian word, considered as an element of $\stab({\cal L}\cup{\cal R})$, has a directive sequence 
containing infinitely many occurrences of elements of ${\cal L}$.
In other words, any episturmian word belongs to $\stab({\cal R}^*{\cal L})$.
\end{lemma}

\begin{proof}
Let $\bw$ in $\stab({\cal L}\cup{\cal R})$
and let $(\bw_n)_{n \geq 0}$ and $(\sigma_n)_{n \geq 1}$ as in Lemma~\ref{L:normalization}.
Observe that, for any infinite word $\bu$
and any letters $\alpha$, $\beta$,
the word $R_\beta(\bu)$ begins with $\alpha$
if and only if $\bu$ begins with $\alpha$.

Assume that $(\sigma_n)_{n \geq \ell+1} \in {\cal R}^\omega$ for some integer $\ell$.
Let $\alpha$ be the first letter of $\bw_\ell$.
{From the previous observation, it follows that $\alpha$ is the first letter of $\bw_k$ for all $k \geq \ell$.
The hypothesis
``if $\bw_k$ begins with the letter $\alpha$, then $\sigma_{k+1} \neq R_\alpha$",
implies that $\sigma_k \neq R_\alpha$ for all $k \geq \ell+1$.
Hence, for any $k \geq \ell$, $\alpha$ occurs only as a prefix of the words $\sigma_\ell\cdots\sigma_k(\alpha)$ whose lengths grow to infinity.
As these words $\sigma_\ell \cdots \sigma_k(\alpha)$ are prefixes of $\bw_\ell$,}
$\alpha$ has only one occurrence in $\bw_\ell$ and the word $\bw_\ell$ is not recurrent.
Using Lemma~\ref{L:recurrence}
and the previously recalled hypothesis,
by inverse induction,
we can show that $\bw_i$ is also not recurrent for every $i$, $\ell \geq i \geq 0$.
Hence, by Proposition~\ref{P:another_characterization_of_episturmian_words},
$\bw_0$ is not episturmian. 
Thus, 
for all $\ell \geq 1$, there exists $k \geq \ell$ such that $\sigma_k \in {\cal L}$, {that is, $\bw \in \stab({\cal R}^*{\cal L})$}.
\end{proof}

\begin{proposition}
\label{P:stab_epi}
{Consider an alphabet containing at least two letters.}
The set ${\cal E}$ of episturmian words is the set $\stab({\cal R}^*{\cal L})$ and 
it is also the set
of all {${\cal R}^*{\cal L}$}-adic words.
\end{proposition}

\begin{proof}
From {Theorem~\protect\ref{T:charac_episturmien_JP}} and Lemma~\ref{L:form_direct_episturmian}, 
${\cal E} \subseteq \stab({\cal R}^*{\cal L})$.

Assume now that $\bw \in \stab({\cal R}^*{\cal L})$.
Let $(\sigma_n)_{n \geq 1}$ be a directive sequence of $\bw$ ($\sigma_n \in {\cal R}^*{\cal L}$ for all $n \geq 1$).
If there exists a letter $a$ such that $\sigma_n \in R_a^*L_a$ for arbitrarily large $n$,
then,
$\bw = f(a^\omega)$ for some $f$ in {$({\cal R}^*{\cal L})^*$}.
{As there exists a letter $b \neq a$, $a^\omega = \lim_{n \to \infty} L_a^n(b)$: $\bw$ is  ${\cal R}^*{\cal L}$-adic}.
Otherwise,
let $(\bw_n)_{n \geq 0}$ be a sequence of desubstituted words of $\bw$ using $(\sigma_n)_{n \geq 1}$.
For any integer $n \geq 0$,
there exists $m$, $m \geq n$,
such that $|\sigma_{n+1}\cdots \sigma_m(\first(\bw_m))| \geq 2$.
Hence $\bw$ is $({\cal R}^*{\cal L})$-adic.

Assume now that $\bw$ is {${\cal R}^*{\cal L}$-adic}. 
{A first case is that}
$\bw = f(a^\omega)$ for some letter $a$ and some $f$ in $({\cal L}\cup {\cal R})^*$.
{In this case} $\bw$ is recurrent as any periodic word is.
Moreover it belongs to $\stab({\cal L}\cup{\cal R})$ 
with  $fL_a^\omega$ as a directive sequence.
By Proposition~\ref{P:another_characterization_of_episturmian_words}, $\bw$ is episturmian.

{Assume now that the first case does not hold (but still $\bw$ is ${\cal R}^*{\cal L}$-adic). 
By Proposition~{\protect\ref{S-adic=>stable}},
$\bw \in \stab({\cal R}^*{\cal L})$.
Let $\bs = (\sigma_n)_{n \geq 1}$ be the directive sequence of $\bw$ as an element of this stable set
(for all $n \geq 0$, $\sigma_n$ belongs to $\stab({\cal R}^*{\cal L}))$
and let $(\bw_n)_{n \geq 0}$ be the sequence of desubstituted words.
Observe that there cannot exist a letter $\alpha$ and an integer $N$
such that, for all $n \geq N$,  $\sigma_n \in (R_\alpha^*L_\alpha)$.
Indeed, for any infinite word $\bu$ and any element $f$ of $(R_\alpha^*L_\alpha)^k$, $f(\bu)$ begins with $\alpha^k$.
Hence the existence of $\alpha$ and $N$ would implies that $\bw_N = \alpha^\omega$ and that we are back in the first case.
From the previous observation, we deduce that,}
for all $n \geq 0$,
there exist an integer $m \geq 1$ 
and a letter $b \neq \first(\bw_n)$ such that
 $|\sigma_{n+1}\cdots \sigma_m(\first(\bw_m))|_b \neq 0$.
 This implies that, for all $n \geq 0$,
 $\first(\bw_n)$ has at least two occurrences in $\bw_n$.
 Observe that any factor of $\bw$ is a factor of a prefix of $\bw$
 and so a factor of $\sigma_1\cdots \sigma_n(\first(\bw_n))$ for some integer $n$.
 Hence any factor of $\bw$ occurs at least twice in $\bw$: $\bw$ is recurrent.
 By Proposition~\ref{P:another_characterization_of_episturmian_words}, $\bw$ is episturmian.

\end{proof}

As ${\cal R}^*{\cal L}$ is infinite, the next proposition answers a natural question.

\begin{proposition}
\label{P:stab_epi_pas_fini}
{Assume that the alphabet $A$ contains at least two letters.}
There is no finite set $\cS \subseteq \subst(A)$
such that the set ${\cal E}$ of episturmian words is $\stab(\cS)$.
\end{proposition}

For proving this result,
we need a characterization of the set of endomorphisms preserving episturmian words.

Let ${\cal P}$ be the set of all endomorphisms $f$
for which there exists a letter $a$ such that{, for all letter $b$ in $A$, $f(b) \in a^+$}.
For any infinite word $\bw$ and $f \in {\cal P}$, $f(\bw)$ is periodic: it is the word $a^\omega$ for some letter $a$.
Finally{,} let
$\cS_{\rm epi} = ({\cal L} \cup {\cal R} \cup \exch \cup {\cal P})^*$.

\begin{proposition}
\label{P:preserving_episturmian}
The set of endomorphisms preserving ${\cal E}$ is $\cS_{\rm epi}$.
\end{proposition}

\begin{proof}
Let us first verify that any element of $\cS_{\rm epi}$ preserves ${\cal E}$.
It is sufficient to consider elements of ${\cal L} \cup {\cal R} \cup \exch \cup {\cal P}$.
As for any letter $\alpha$ the word $\alpha^\omega$ is episturmian,
this holds for elements of ${\cal P}$.
Changing the alphabet of a word does not change the property of being episturmian. 
Hence this holds also for elements of $\exch$.
Finally{,} as ${\cal E} = \stab({\cal R}^*{\cal L})$ {by Proposition~{\protect\ref{P:stab_epi}}}, elements of $({\cal L}\cup{\cal R})$
also preserve ${\cal E}$.

To prove {that any endomorphisms preserving ${\cal E}$ is an element of $\cS_{\rm epi}$,} we need the next two lemmas.

\begin{lemma}
\label{L:reduce_preserve}
{Let $f$, $g$ in $\subst(A)$ and $\alpha \in A$ such that $f = L_\alpha g$ or $f = R_\alpha g$.}
{We have:}
$f$ preserves ${\cal E}$
if and only if
$g$ preserves ${\cal E}$.
\end{lemma}

\begin{proof}
We have just seen that elements of ${\cal L}\cup{\cal R}$ preserve ${\cal E}$.
Hence if $g$ preserves ${\cal E}$, $f$ also preserves ${\cal E}$.

From now on,
 assume that $f$ preserves ${\cal E}$ but not $g$.
This means that there exists $\bw$ in ${\cal E}$
such that $g(\bw) \not\in {\cal E}$ and $f(\bw) \in {\cal E}$.
Observe that $g(\bw)$ is recurrent as $\bw$ is recurrent.
By definition of episturmian words,
$g(\bw)$ has two left special factors of the same length
or its set of factors is not closed by reversal.

First assume the existence of two different left special factors.
Considering them of minimal length
we may assume that these words are $ua_1$ and $ua_2$ for a word $u$ and different letters $a_1$ and $a_2$.
There exist letters $b_1$, $b_2$, $c_1$, $c_2$ with $b_1 \neq b_2$, $c_1 \neq c_2$
such that $b_1 u a_1$, $b_2 u a_1$, $c_1 u a_2$, $c_2 u a_2$ 
are factors of $g(\bw)$.
As $g(\bw)$ is recurrent,
we can consider non-prefix occurrences 
of the previous words.

{If $f = L_\alpha g$, then  the words 
$b_1L_\alpha(u)\alpha a_1$, 
$b_2L_\alpha(u)\alpha a_1$, 
$c_1L_\alpha(u)\alpha a_2$, 
$c_2L_\alpha(u)\alpha a_2$
are factors of $L_\alpha(g(\bw))$. 
To verify this assertion it may be observed that
(for instance) $b_1L_\alpha(u)\alpha a_1$ is a factor of 
$L_\alpha(b_1ua_1)\alpha$ which is a factor of $L_\alpha(g(\bw))$ 
(note that moreover $a_1\alpha = \alpha a_1$ when $a_1 = \alpha$).}

{If $f = R_\alpha g$, then the words
$b_1\alpha R_\alpha(u) a_1$, 
$b_2\alpha R_\alpha(u) a_1$, 
$c_1\alpha R_\alpha(u) a_2$, 
$c_2\alpha R_\alpha(u) a_2$
are factors of $R_\alpha(g(u))$. 
To verify this assertion it may be observed that
(for instance) $b_1\alpha R_\alpha(u) a_1$ is a factor of
$\alpha R_\alpha(b_1 u a_1)$ which is also a factor of $R_\alpha(g(\bw))$ as we consider a non-prefix occurrence of $b_1ua_1$.}

{In both cases,} $f(\bw)$ has two left special factors of the same length
and consequently it is not episturmian. A contradiction.

From now on{,} assume that $g(\bw)$ contains a factor $u$
but not its reversal $\tilde{u}$.
The word $L_\alpha(g(\bw))$ contains the factor $L_\alpha(u)\alpha$ but not its reversal.
Indeed if it contains $\widetilde{L_\alpha(u)\alpha} = \alpha\widetilde{L_\alpha(u)}
= \alpha R_\alpha(\tilde{u}) = L_\alpha(\tilde{u})\alpha$, then
the word $g(\bw)$ contains the factor $\tilde{u}$: a contradiction.
Similarly the word $R_\alpha(g(\bw))$ contains the factor $\alpha R_\alpha(u)$ 
(remember that $g(\bw)$ is recurrent and so there exists a non-prefix occurrence of $u$)
but not its reversal $\alpha R_\alpha(\tilde{u})$.
{As $f(\bw) = L_\alpha(g(\bw))$ or
$f(\bw) = R_\alpha(g(\bw))$, $f(\bw)$ is not episturmian.
This contradicts the facts that $\bw$ is episturmian and that $f$ preserves ${\cal E}$.}
Hence $g$ preserves ${\cal E}$.
\end{proof}

\begin{lemma}
\label{L:begin_and_end}
Let $f \in \subst(A)\setminus \exch^*$ be a morphism that preserves ${\cal E}$.
All images of letters by $f$ begin with the same letter,
or, all images of letters end with the same letter.
\end{lemma}

\begin{proof}
Let $\last(w)$ be the last letter of a non-empty {finite} word $w$. 

{Let $f \in \subst(A)\setminus \exch^*$ be a morphism that preserves ${\cal E}$.}
{We start with the following observation: 
For any episturmian word $\bw$, there exists a letter $\alpha$ such that any factor of $\bw$ of length 2 contains $\alpha$.}
{This follows directly from Theorem~{\protect\ref{T:charac_episturmien_JP}} that states, 
that for any episturmian word $\bw$, there exists a letter $\alpha$
such that $\bw = L_\alpha(\bw')$ or $\bw = R_\alpha(\bw')$ for some episturmian word $\bw'$.}

\textbf{Case 1}.
Assume first that {there exist letters $\alpha$ and $a$ such that}
$\alpha\alpha$ is a factor of $f(a)$.
Then{, by the previous observation,} for any episturmian word $\bw$ containing $a$,
$f(\bw)$ cannot contain a factor $\beta \gamma$ with $\beta$, $\gamma$ 
two letters with $\beta \neq \alpha$ and $\gamma \neq \alpha$.
This implies that there cannot exist letters $b$ and $c$
such that $\first(f(c)) \neq \alpha$ and
$\last(f(b)) \neq \alpha$.
Indeed there exist episturmian words containing both $a$ and $bc$: their images by $f$ would contain both $\alpha\alpha$ and 
the word $ \last(f(b))\first(f(c))$ which is impossible.
Thus all images of letters begin with $\alpha$ or
 all images of letters end with $\alpha$.
 
 \textbf{Case 2}.
Assume now {that the previous case does not hold but that there exist pairwise distinct letters $\alpha$, }
such that 
$\alpha\beta$ or $\beta\alpha$ is a factor of an element of $f(A)$,
and, $\alpha\gamma$ or $\gamma\alpha$ is also a factor of an element of $f(A)$.

{Assume first that there exist an episturmian word $\bw$ containing all letters, and, two letters $\delta_1$ and $\delta_2$ different from $\alpha$
such that $\delta_1\delta_2$ is a factor of $f(\bw)$.
Then by initial observation, as the previous case does not hold, $\delta_1\delta_2 \in \{\beta\gamma, \gamma\beta\}$.
It follows that no other letter than $\alpha$, $\beta$ and $\gamma$ occurs in words of $f(A)$.
Indeed otherwise, $f(\bw)$ would contain two factors written on different letters: 
a contradiction with the initial observation.
Assume now by contradiction that there exist images of letters that do not begin with the same letter
and images of letters that do not end with the same letter.
Then there exist a letter $x \in \{\alpha, \beta, \gamma\}$
such that $x = \first(f(b))$ and $x = \last(f(a))$ for some letters $a$ and $b$.
Let $\bw$ be an episturmian word containing all letters and the factor $ab$.
The episturmian word $f(\bw)$ contains $xx$ and a factor of length 2 that does not contains $x$.
This contradicts the initial observation.}

Hence, for any $\bw \in {\cal E}$, 
$f(\bw)$ cannot contain a factor $\delta_1\delta_2$ 
with $\delta_1$, $\delta_2$ two letters different from $\alpha$.
We end as in the Case~1.

\textbf{Case 3}.
Assume now that the previous cases do not hold but 
{that there exists a letter $b$ such that $|f(b)| \geq 2$.}

{Note that, as the two previous cases do not hold, 
there must exist distinct letters $\alpha$ and $\beta$ such that
factors of length 2 of 
elements of $f(A)$ are necessarily $\alpha \beta$ and $\beta\alpha$.
This implies that the first letter of $f(b)$ is $\alpha$ or $\beta$.
Without loss of generality, assume $f(b)$ begins with $\alpha$.
Hence, for some $k \geq 1$, 
$f(b) = (\alpha\beta)^k$ of
$f(b) = (\alpha\beta)^k\alpha$. }

{\textit{Case $f(b) = (\alpha\beta)^k$}}. 
There cannot exist letters $a$, $c$, $\gamma$, $\delta$
such that $\gamma \neq \beta$, $\alpha \neq \delta$,
$\gamma = \last(f(a))$,
$\delta = \first(f(c))$.
Indeed otherwise $f(abc)$ would contain the factor
$\gamma(\alpha\beta)^k\delta${.
Since} $\beta$ {is} not a factor of $\gamma\alpha$ and $\alpha$ {is} not a factor of {$\beta\delta$},
for any episturmian word $\bw$ containing $abc$,
the word $f(\bw)$ would not be {an episturmian word
by the initial observation.
Hence all images of letters begin with $\alpha$, or, 
all images of letters end with $\beta$.}

{\textit{Case $f(b) = (\alpha \beta)^k \alpha$}.}
Assume that there exists a letter $a$ such that $\last(f(a)) = \gamma$ with $\gamma \neq \alpha$.
{As all length $2$ factors of images of letters are $\alpha\beta$ and $\beta\alpha$,}
if $\gamma \neq \beta$, then $f(a) = \gamma$.
In this case, for any episturmian word containing $aab$,
$f(\bw)$ contains both $\gamma\gamma$ and $\alpha\beta$
contradicting the fact that it is episturmian.
So $\gamma = \beta$.
{Moreover $f(a) =  (\alpha\beta)^\ell$ or
$f(a) = (\beta \alpha)^\ell\beta$ for some $\ell \geq 0$.
The case $f(a) =  (\alpha\beta)^\ell$ has already been considered.
Hence assume $f(a) = (\beta \alpha)^\ell\beta$.}
Let $\bw$ be an episturmian word containing the factor $bbbabb$.
The word $f(\bw)$ contains the factor
$$\alpha(\alpha\beta)^k\alpha(\alpha\beta)^k\alpha(\beta\alpha)^\ell\beta(\alpha\beta)^k\alpha\alpha$$
which is equals to
$$\alpha(\alpha\beta)^k\alpha(\alpha\beta)^{k+\ell+1+k}\alpha\alpha$$
As {$f(\bw) \in \stab({\cal L}\cup{\cal R})$ by Theorem~\protect\ref{T:charac_episturmien_JP}},
$f(\bw) = L_\alpha(\bw')$ or $f(\bw) = R_\alpha(\bw')$ for some episturmian word $\bw'$.
Thus $\bw'$ should contain the factor 
$\alpha\beta^k\alpha\beta^{k+\ell+1+k}\alpha$.
{As ${k+\ell+1+k} \geq k+2$, this implies that both $\beta^{k+1}$ and $\beta^k\alpha$ are left special factors of $\bw'$: 
this contradicts the fact that $\bw'$ is episturmian.}
Thus for each letter $a$, {$\last(f(a)) = \alpha$}.

{\textbf{Case 4}.}
It remains to consider the case where
all images of letters are of length 1.
As $f \not\in \exch^*$,
if all images are not the same,
then there exist pairwise different letters $a$, $b$ and $c$
such that $f(a) = f(b)$
and $f(a) \neq f(c)$.
Set {$x = f(a)$}, {$y = f(c)$} ($x$ and $y$ are letters).
Let $\bw$ be an episturmian word containing the factor $acc$.
The word $L_aL_bL_c(\bw)$ contains the factor
$cabaabacabac$ and
the word $fL_aL_bL_c(\bw)$ contains the factor
$yx^6yx^3y$.
{Hence $x^3y$ and $x^4$ are both left special factors of $fL_aL_bL_c(\bw)$.}
This is not possible for an episturmian word: we have a contradiction with the fact
that $f$ preserves ${\cal E}$.
Thus all images of letters are equal.
\end{proof}

Let us continue the proof of Proposition~\ref{P:preserving_episturmian}.

From now on, let $f$ be an endomorphism preserving episturmian words.
Let $||f|| = \sum_{\alpha \in A} |f(\alpha)|$. We prove by induction on $||f||$
that $f \in \cS_{\rm epi}$.

Assume first that $||f||= \#A$.
If $f(A) = A$,
then $f \in \exch^* \subseteq {\cal S}_{epi}$.
Otherwise by Lemma~\ref{L:begin_and_end}, $f \in {\cal P} \subseteq  {\cal S}_{epi}$.

From now on{,} assume that $||f|| > \#A$
and assume $f \not\in {\cal P}$.
By Lemma~\ref{L:begin_and_end},
all images of letters begin with the same letter,
or, all images of letters end with the same letter.
Assume that the first case holds (the second case is symmetric) and let $\alpha$ be the first letter of images of letters.
Let $\bw \in {\cal E}$.
By hypothesis $f(\bw) \in {\cal E}$.
As ${\cal E} \subseteq \stab({\cal L}\cup{\cal R})$ by Theorem~\ref{T:charac_episturmien_JP},
$f(\bw) = L_\beta(\bw')$ or $f(\bw) = R_\beta(\bw')$ for some letter $\beta$ and some word $\bw'$.

When $f(\bw) = L_\beta(\bw')$, we have $\alpha = \beta$ and
we can find a morphism $g$ such that $f = L_\alpha g$.
By Lemma~\ref{L:reduce_preserve}, $g$ preserves ${\cal E}$.
As $f \not\in {\cal P}$, $||g|| < ||f||$ and by induction $g \in \cS_{\rm epi}$. So $f \in \cS_{\rm epi}$.

Assume now that $f(\bw) = R_\beta(\bw')$.
If all images of letters by $f$ end with $\beta$,
then $f = R_\beta g$ and, as in the case $f = L_\alpha g$ above, $f \in \cS_{\rm epi}$.
If some image of {a letter} by $f$ {does} not end with $\beta$,
as {the} image of letters is followed by $\alpha$ in $f(\bw)$
and by $\beta$ in $R_\beta(\bw')$, we get $\alpha = \beta$.
Hence $f = L_\alpha h$ for some morphism $h$, and as previously $f \in \cS_{\rm epi}$.
\end{proof}

\medskip

\begin{proof}[Proof of Proposition~\ref{P:stab_epi_pas_fini}]
{The proof is similar to the proof of the second part of 
Proposition~{\protect{\ref{P:stable_standard_episturmian}}} but we have to take care of elements in ${\cal R}$.}
Assume {by contradiction} that there exists a finite set $\cS \subseteq \subst(A)$
such that the set of episturmian words is $\stab(\cS)$.
By Remark~\ref{Rem:preserve},
elements of $\cS$ preserve ${\cal E}$ and so,
by Proposition~\ref{P:preserving_episturmian}, {$\cS \subseteq \cS_{\rm epi} = ({\cal L} \cup {\cal R} \cup \exch \cup {\cal P})^*$}.

Observe that $\cS \cap \exch^* = \emptyset$.
Indeed if {$f \in \cS \cap \exch^*$},
$f$ acts on the alphabet as a permutation 
and there exists an integer $n$ such that $f^n$ is the identity.
Then any infinite word belongs to $\stab(\cS)$: a contradiction.

From now on, we consider only aperiodic episturmian words.
These words {can} be desubstituted only on elements of
$\cS \cap ({\cal L}\cup{\cal R}\cup \exch)^*$.
{From Relations~{\protect{\eqref{eq1}}} to {\protect{\eqref{eq4}}},} 
any element of $({\cal L}\cup{\cal R}\cup \exch)^*$ can be decomposed
into $\sigma e$ with $\sigma \in ({\cal L}\cup{\cal R})^*$ and $e \in  \exch^*$,
that is, it can be viewed as an element of $({\cal L}\cup{\cal R})^*\exch^*$.

Let $a$ and $b$ be two different letters. The following fact is important.

{
\textbf{Fact.} Let $i$, $j$ be integers such that $i \leq j$. 
Let $\bw \in A^\omega$ and $f \in ({\cal L}\cup{\cal R})^i$. 
If $f(\bw)$ begins with $ab^{j+2}a$,
then $f = R_b^i$.
}

{
We prove this fact by induction on $i$. It is basically true for $i = 0$.
Assume $i \geq 1$.
As $f(\bw)$ begins with $abb$ and $f \in ({\cal L}\cup{\cal R})^i$,
necessarily $f = R_b g$ for some $g \in ({\cal L}\cup{\cal R})^{i-1}$.
Moreover $g(\bw)$ has a prefix $ab^{(j-1)+2}a$ and $i-1 \leq j-1$.
Hence $g = R_b^{i-1}$ by induction. So $f = R_b^i$.
}

{
Now let $M_1 = \max( \{ |f| \mid f\pi \in \cS, f \in ({\cal L}\cup{\cal R})^*, \pi \in \exch^*\}$. 
(here, for $f \in ({\cal L}\cup{\cal R})^*$, $|f|$ is the length of $f$ considered as a word over the alphabet ${\cal L}\cup{\cal R}$).
Let
$M = (\#A!+1)M_1$.
Let $\bw$ be an episturmian word having $ab^{M+2}a$ as a  prefix.
Let $(\sigma_n)_{n \geq 1} \in \cS^\omega$ be a directive word of $\bw$ (as an element of $\stab(\cS)$).
}

{
Let $i$, $1 \leq i \leq \#A!+1$. 
For $1 \leq j \leq i$, as $\sigma_j \in  ({\cal L}\cup{\cal R})^*\exch^*$,
there  exist $g_j \in ({\cal L}\cup{\cal R})^*$ and $\pi_j' \in \exch^*$
such that $\sigma_j = g_j \pi_j'$.
Using Equations~{\protect{\eqref{eq2}}} and {\protect{\eqref{eq4}}}, 
we can see that 
there exist $f_i \in ({\cal L}\cup{\cal R})^*$ and $\pi_i \in \exch^*$ such that $\sigma_1\cdots\sigma_i = f_i \pi_i$.
Moreover $|f_i| = \sum_{j = 1}^i |g_j|$ and so $|f_i| \leq i M_1 \leq M$.
From the previous fact, $f_i \in\{R_b\}^*$.
Since the cardinality of $\exch^*$ is $\#A!$, there exists $i$ and $j$, $1 \leq i < j \leq \#A!+1$, such that $\pi_i = \pi_j$. Set $\pi = \pi_i$.
}

{
Let $\ell_i$ and $\ell_j$ be the integers such that $\sigma_1 \ldots \sigma_i = R_b^{\ell_i} \pi$ 
and $\sigma_1 \ldots \sigma_j = R_b^{\ell_j}\pi$. We have $R_b^{\ell_j}\pi =$ 
$\sigma_1 \ldots \sigma_j =$ 
$R_b^{\ell_i}\pi\sigma_{i+1}\cdots\sigma_j$. Thus 
$\pi\sigma_{i+1}\cdots\sigma_j = R_b^{\ell_j-\ell_i}\pi$.
By Equations~{\protect{\eqref{eq3}}} and {\protect{\eqref{eq4}}}, 
there exists a letter $c$ such that
$R_b^{\ell_j-\ell_i}\pi = \pi R_c^{\ell_j-\ell_i}$. So $\sigma_{i+1}\cdots\sigma_j =$ $R_c^{\ell_j-\ell_i}$.
}

Observe that, for $d$ a letter different from $c$ (the alphabet $A$ contains at least two letters), $R_c^\omega$ is a directive sequence of the word $dc^\omega$.
{Thus $(R_c^{\ell_j-\ell_i})^\omega$ is a directive sequence of the word $dc^\omega$.}
Hence this word $dc^\omega$ belongs to $\stab(\cS)$ although it is not episturmian ({by Theorem~{\protect\ref{T:charac_episturmien_JP}}} since it is not recurrent).
We have obtained a contradiction with $\stab(\cS) = {\cal E}$.
\end{proof}

\subsection{Strict episturmian words}
\noindent

Theorem~\ref{T:charac_episturmien_JP} recalls only a part of Theorem 3.10 in \cite{JustinPirillo2002}.
This latter theorem also implies that
an infinite word $\bw$ is $A$-strict episturmian if and only if
there exist an infinite sequence of \textit{recurrent} infinite
words $(\bw_n)_{n \geq 0}$
and a sequence $(\sigma_n)_{n \geq 1}$ in $({\cal L}\cup {\cal R})^\omega$
such that, for each letter $\alpha$ in $A$, 
$L_\alpha$ or $R_\alpha$ occurs infinitely often in $(\sigma_n)_{n \geq 1}$.
Let ${\cal S}_{strictepi}$ be the set 
$({\cal L}\cup{\cal R})^*{\cal L}({\cal L}\cup{\cal R})^*\cap\cap_{\alpha \in A} ({\cal L}\cup {\cal R})^* \{L_\alpha, R_\alpha\} ({\cal L}\cup {\cal R})^*$
of all elements of $({\cal L}\cup{\cal R})^*$ having a decomposition over ${\cal L}\cup{\cal R}$ with at least one element of
${\cal L}$ and at least one element of $\{L_\alpha, R_\alpha\}$ for each letter $\alpha$.

From what precedes, using Lemma~\ref{L:normalization},
we can see that the set of $A$-strict episturmian words
is included in $\stab({\cal S}_{strictepi})$.
Conversely, acting as in the proof of Proposition~\ref{P:stab_epi},
we can deduce that any element of 
$\stab({\cal S}_{strictepi})$
is recurrent. The next result follows.

\begin{proposition}
The set of $A$-strict episturmian words is $\stab({\cal S}_{strictepi})$.
\end{proposition}

In \cite[Th. 3.13]{JustinPirillo2002},
it is proved that $({\cal L}\cup {\cal R}\cup Exch)^*$
is the set of endomorphisms of $A^*$
that preserve $A$-strict episturmian words.
Using this result, as done for the proof of Proposition~\ref{P:stab_epi_pas_fini}, we can prove:

\begin{proposition}
There is no finite set $\cS \subseteq \subst(A)$
such that $\stab(\cS)$ is the set of $A$-strict episturmian words.
\end{proposition}

\section{\label{sec:conclusion}Conclusion}

Stable sets formalize the concept of infinite desubstitutions using a set of nonerasing endomorphisms.
We have shown that several known sets of words are stable sets:
the set of binary balanced words, 
the set of Sturmian words,
the set of Lyndon Sturmian words,
the set of standard episturmian words 
({which} corresponds, in the binary case, to the set of LSP words),
the set of strict standard episturmian words
({which} corresponds, in the binary case, to the set of standard words),
the set of episturmian words
and the set of strict standard episturmian words.
Among all these sets,
only the set of binary balanced words and
the set of standard episturmian words are
stable sets of a finite set of substitutions.
A first natural question is whether 
there exist other sets defined by combinatorial properties that are stable sets
of a (finite) set of substitutions.
 
 {A characterization} of a set of words as {the} stable set of an infinite set
 of substitutions may be more difficult to understand and to use
 than a characterization as a subset of a stable set of a
 finite set of substitutions using conditions on directive sequences.
 For instance, it is probably more interesting to know that
 standard Sturmian words are the elements of $\stab({\cal L}\cup {\cal R})$ whose directive
 sequences contain infinitely many occurrences of $L_a$ and 
infinitely many occurrences of $L_b$
than to know that they are the elements of $\stab(\cS_{stand})$,
even if this latter formulation states the same result in a more compact form.
Similarly, 
it may be more interesting to know that episturmian words
are the recurrent elements of  $\stab({\cal L}\cup {\cal R})$ 
than to know that they are the elements of {$\stab({\cal R}^* {\cal L})$}.
More generally one can search for a characterization of 
a set of words as a subset of a stable set whose directive sequence of elements
verify a particular condition. Such a result was obtained for the set of LSP words \cite{Richomme2019IJFCS}.
Moreover this approach is often done {w.r.t.} the concept of $S$-adicity instead of stable sets (see, \textit{e.g.} the case 
of Sturmian words or the paper \cite{Berthe_Delecroix2014RIMS} and its references).

Some extensions of the notion of stable set could also be studied.
For instance it should be quite natural to 
search for set{s} of words that are images by a morphism or by a set of morphisms
(not necessarily endomorphisms) of a stable set. 
But except the morphic words, 
the author knows no example among classical sets of words.

Another direction of study could be to have
a better formalization of the possible changes of alphabets.
Indeed{,} remember that in the definitions of $S$-adicity,
the considered morphisms are not necessarily endomorphisms.
Also in the examples of stable sets
one can observe that some elements of stable sets have desubstituted words written on alphabets
whose cardinalities may decrease (see for instance the word $L_cL_b(a^\omega)$).

\section*{{Acknowledgements}}
Many thanks to Robert Merca\c{s} and Daniel Reidenbach for their invitation to talk to the Words 2019 conference that stimulated the writing of this paper.


\small

\begin{thebibliography}{10}

\bibitem{ArnouxMizutaniSellami2014TCS}
P.~Arnoux, M.~Mizutani, and T.~Sellami.
\newblock Random product of substitutions with the same incidence matrix.
\newblock {\em Theor. Comput. Sci.}, 543:68--78, 2014.

\bibitem{Berthe_Delecroix2014RIMS}
V.~Berth\'e and V.~Delecroix.
\newblock Beyond substitutive dynamical systems: S-adic expansions.
\newblock In S.~Akiyama, editor, {\em Numeration and Substitution 2012}, volume
  B46 of {\em RIMS K\^oky\^uroku Bessatsu}, pages 81--123, 2014.

\bibitem{BertheHoltonZamboni2006}
V.~Berth\'e, C.~Holton, and L.~Q. Zamboni.
\newblock Initial powers of {S}turmian sequences.
\newblock {\em Acta Arith.}, 122:315--347, 2006.

\bibitem{Berthe_Rigo2010CANT}
V.~Berth\'e and M.~Rigo, editors.
\newblock {\em Combinatorics, Automata and Number Theory}, volume 135 of {\em
  Encyclopedia of Mathematics and its Applications}.
\newblock Cambridge University Press, 2010.

\bibitem{DroubayJustinPirillo2001}
X.~Droubay, J.~Justin, and G.~Pirillo.
\newblock Episturmian words and some constructions of {de Luca} and {Rauzy}.
\newblock {\em Theoret. Comput. Sci.}, 255:539--553, 2001.

\bibitem{Ferenczi1996ETDS}
S.~Ferenczi.
\newblock Rank and symbolic complexity.
\newblock {\em Ergodic Theory Dynam. Systems}, 16:663--682, 1996.

\bibitem{Fici2011TCS}
G.~Fici.
\newblock Special factors and the combinatorics of suffix and factor automata.
\newblock {\em Theoret. Comput. Sci.}, 412:3604--3615, 2011.

\bibitem{GlenLeveRichomme2008TIA}
Amy Glen, Florence Lev{\'e}, and Gw{\'e}na{\"e}l Richomme.
\newblock Directive words of episturmian words: equivalences and normalization.
\newblock {\em Theor. Inform. Appl.}, 43(2):299--319, 2009.

\bibitem{Godelle2010AAM}
E.~Godelle.
\newblock The stable set of a self-map.
\newblock {\em Adv. in Appl. Math.}, 45:438--448, 2010.

\bibitem{JustinPirillo2002}
J.~Justin and G.~Pirillo.
\newblock Episturmian words and episturmian morphisms.
\newblock {\em Theoret. Comput. Sci.}, 276(1-2):281--313, 2002.

\bibitem{Leroy2012thesis}
J.~Leroy.
\newblock Contribution \`a la r\'esolution de la conjecture {$S$}-adique.
\newblock Doctoral Thesis, Universit\'e de Picardie Jules Verne, 2012.

\bibitem{Leroy2012AAM}
J.~Leroy.
\newblock Some improvements of the {$S$}-adic conjecture.
\newblock {\em Adv. App. Math.}, 48:79--98, 2012.

\bibitem{Leroy2014DMTCS}
J.~Leroy.
\newblock An {$S$}-adic characterization of minimal subshifts with first
  difference of complexity $p(n+1)-p(n)\le2$.
\newblock {\em Discrete Math. Theor. Comput. Sci.}, 16(1)(1):233--286, 2014.

\bibitem{Leroy_Richomme2012integers}
J.~Leroy and G.~Richomme.
\newblock A combinatorial proof of {S-adicity} for sequences with linear
  complexity.
\newblock {\em Integers}, 13(Article \#A5), 2013.

\bibitem{LeveRichomme2007TCS}
F.~Lev{\'e} and G.~Richomme.
\newblock Quasiperiodic {S}turmian words and morphisms.
\newblock {\em Theor. Comput. Sci.}, 372(1):15--25, 2007.

\bibitem{Lothaire1983book}
M.~Lothaire.
\newblock {\em Combinatorics on Words}, volume~17 of {\em Encyclopedia of
  Mathematics and its Applications}.
\newblock Addison-Wesley, 1983.
\newblock Reprinted in the {\em Cambridge Mathematical Library}, Cambridge
  University Press, UK, 1997.

\bibitem{Lothaire2002}
M.~Lothaire.
\newblock {\em Algebraic Combinatorics on Words}, volume~90 of {\em
  Encyclopedia of Mathematics and its Applications}.
\newblock Cambridge University Press, 2002.

\bibitem{MorseHedlund1940}
M.~Morse and G.~A. Hedlund.
\newblock Symbolic dynamics {II}. {S}turmian trajectories.
\newblock {\em Amer. J. Math.}, 62:1--42, 1940.

\bibitem{Pytheas2002}
N.~Pytheas~Fogg.
\newblock {\em Substitutions in Dynamics, Arithmetics and Combinatorics},
  volume 1794 of {\em Lecture Notes in Mathematics}.
\newblock Springer, 2002.
\newblock (V. Berth{\'e}, S. Ferenczi, C. Mauduit, A. Siegel, editors).

\bibitem{Richomme2003BBMS}
G.~Richomme.
\newblock Lyndon morphisms.
\newblock {\em Bulletin of the Belgian Mathematical Society}, 10(5):761--786,
  2003.

\bibitem{Richomme2019IJFCS}
G.~Richomme.
\newblock Characterization of infinite {LSP} words and endomorphisms preserving
  the {LSP} property.
\newblock {\em Internat. J. Found. Comput. Sci.}, 30(1):171--196, 2019.

\bibitem{Thue1912}
A.~Thue.
\newblock Uber die gegenseitige {L}age gleigher {T}eile gewisser
  {Z}eichenreihen.
\newblock {\em Kristiania Videnskapsselskapets Skrifter Klasse I. Mat.-naturv},
  1:1--67, 1912.

\end{thebibliography}
\end{document}